\documentclass[draftcls,onecolumn]{IEEEtran}
\usepackage[T1]{fontenc}
\usepackage[utf8]{inputenc}
\usepackage{amsmath,amssymb,amsthm,mathtools,bm}
\usepackage[shortlabels]{enumitem}
\usepackage{xcolor}
\usepackage{hyperref}
\usepackage{dsfont}
\usepackage[nameinlink,noabbrev]{cleveref}
\usepackage{setspace}

\hypersetup{
    colorlinks=true,
    linkcolor=blue,
    filecolor=magenta,
    urlcolor=blue,
    citecolor=blue,
}

\newtheorem{theorem}{Theorem}

\newtheorem{lemma}{Lemma}
\newtheorem{corollary}{Corollary}
\theoremstyle{definition}
\newtheorem{definition}{Definition}
\newtheorem{remark}{Remark}

\DeclareMathOperator{\Prb}{\mathbb{P}}
\DeclareMathOperator{\E}{\mathbb{E}}

\DeclareMathOperator{\supp}{supp}
\newcommand{\numberthis}{\addtocounter{equation}{1}\tag{\theequation}}

\setlist{nosep}
\title{Classical Commitment over Quantum Channels with Limited  Entanglement Assistance}
\author{R\'{e}mi A. Chou%
\thanks{R. A. Chou is with the Department of Computer Science and Engineering,
The University of Texas at Arlington, Arlington, TX 76019 USA
(e-mail: remi.chou@uta.edu).}
}
\allowdisplaybreaks
\begin{document}
\maketitle

\begin{abstract}
We study classical string commitment over quantum channels with limited
preshared entanglement. For noninteractive protocols, we determine the
commitment capacity of a class of channels with  input dimension
$d$ that, at each use, sample a pair of classical random variables $(F,Z)$,
apply one of the $d^2$ Heisenberg--Weyl operators indexed by $Z$ to the input, and deliver
the transformed quantum system together with $F$ to the receiver. If $E$ is
the available entanglement rate in bits per channel use, then the capacity is
$\min\{H(Z|F),\log_2d+E\}$. This class of channels encompasses quantum erasure and
depolarizing channels, as well as families of Pauli channels.   Additionally, for interactive protocols, we show that the commitment rate cannot exceed $\log_2d+E$ bits per channel use, so that when
$H(Z|F)\geq\log_2d+E$, interactive communication does not increase the capacity. As a consequence, for interactive protocols, we determine the capacity of the quantum erasure channel.
\end{abstract}

\section{Introduction}
\label{sec:introduction}

Commitment is a two-party cryptographic primitive that was implicitly
introduced in the problem of coin flipping by telephone
\cite{blum1983coin}, and can, for instance, be used to construct
zero-knowledge proofs \cite{goldreich1991proofs} or as a subroutine in
secure-computation protocols \cite{canetti2002universally}. It operates in
two phases. In the commit phase, a sender Alice commits to a value while
keeping it hidden from a receiver Bob. In the reveal phase, Alice discloses
the value and provides information allowing Bob to verify it. Security
requires that Bob accept the value revealed by an honest Alice
(Correctness), that Bob learn nothing about this value before the reveal
phase (Hiding), and that a dishonest Alice be unable to reveal a value
different from the one fixed during the commit phase (Binding).

These requirements cannot be achieved simultaneously with
information-theoretic security using noiseless communication alone.
Replacing noiseless classical communication with noiseless quantum
communication does not remove this limitation
\cite{mayers1997unconditionally,lo1998quantum,d2007reexamination}. One
possible approach is to relax information-theoretic security. For example,
computationally secure commitments have been developed for both classical
information \cite{dumais2000perfectly,unruh2016computationally} and quantum
states \cite{gunn2023commitments}, while cheat-sensitive protocols provide
weaker guarantees by making dishonest behavior detectable with nonzero
probability \cite{hardy2004cheat}. Retaining information-theoretic security
is, however, possible by introducing an additional physical resource or
restriction. Examples include communication constraints imposed by special
relativity \cite{kent1999unconditionally,lunghi2013experimental}, bounded or
noisy quantum storage
\cite{damgaard2008bounded,koenig2012unconditional,ng2012experimental,wu2024string},
or a noisy communication channel
\cite{crepeau1997efficient,winter2003commitment}.  More specifically, for the latter approach, \cite{winter2003commitment} introduced the commitment capacity of a channel,
namely the largest asymptotic number of committed bits per channel use, and
determined the capacity of finite nonredundant discrete memoryless channels.
Subsequent works considered efficient constructions
\cite{imai2006efficient}, Gaussian channels
\cite{nascimento2008commitment}, channels whose noise can be partly
controlled by a dishonest party
\cite{crepeau2020commitment,budkuley2022reverse,budkuley2023possibility,wu2024string},
retractable commitments \cite{chou2023retractable}, compound channels
\cite{budkuley2022commitment,yadav2022commitment}, wiretapped channels
\cite{yadav2024wiretapped}, and multiple-access channels
\cite{chou2025multiuser}. An extension from classical channels to
classical--quantum (cq) channels is also discussed in \cite{winter2003commitment} and
further developed in \cite{hayashi2023commitment}.

In this paper, we consider commitment to a classical string using independent
uses of a quantum channel and authenticated noiseless classical communication. We first consider noninteractive protocols, i.e., Bob is silent during the commit phase, when  Alice and Bob may initially
share a pure state whose Schmidt rank is limited. In particular, if the logarithm of this
rank is at most $nE+o(n)$ for $n$ channel uses, then $E$ is referred to as the available
entanglement rate in bits per channel use. 
Within this communication and entanglement model, security must hold against
dishonest parties using arbitrary quantum memories and operations.
Specifically, hiding is required in trace distance against every dishonest
receiver strategy. For binding, a dishonest Alice's commit strategy must
already determine the string that she can later reveal successfully. More
precisely, for every dishonest commit strategy, there must exist a measurement
of the complete commit-phase record that produces a value $\widetilde S$
such that, regardless of Alice's subsequent reveal strategy, Bob accepts a
value different from $\widetilde S$ only with vanishing probability. This
requirement is analogous to extractable binding  used in simulation-based and universally
composable commitment, where, immediately after the commit phase, a reference string can be extracted
using a measurement chosen independently of Alice's subsequent reveal
strategy
\cite{canetti2001universally,damgaard2013unconditionally}.

Our main result concerns a class of channels with prime-power input dimension
$d$ that, at each use, sample a pair of classical random variables $(F,Z)$,
apply one of the $d^2$ Heisenberg--Weyl operators indexed by $Z$ to the input, and deliver
the transformed quantum system together with $F$ to the receiver. Under a condition stated in terms of the
Fourier transform of the joint distribution of $(F,Z)$, we show that the
 capacity with entanglement rate $E$ is
$\min\{H(Z|F),\log_2d+E\}$. This class of channels includes the quantum erasure
channel, the depolarizing channel, and families of Pauli channels, including dephasing
and bit-flip channels.

Our achievability proof approach is to reduce the quantum channel to a classical channel
and then apply a variation of the classical commitment construction of
\cite{winter2003commitment}. Specifically, for a block of $\ell$ channel
uses assisted by $c$ maximally entangled qudit pairs, we associate each
classical input label with one of $d^{\ell+c}$ mutually orthogonal joint
states of the $\ell$ transmitted qudits and Bob's $c$ entangled qudits.
After receiving the channel outputs, Bob jointly measures these outputs and
his entangled qudits using the projective measurement defined by this
orthogonal basis. Every sequence of applied Heisenberg--Weyl operators
permutes the basis states, so Bob's measurement returns Alice's input label
shifted by a random value determined by $Z^\ell$. Together with $F^\ell$,
this measurement outcome defines the induced classical channel. A
random-coding argument shows that, as $\ell$ grows while $c/\ell$ remains
asymptotically fixed, the encoding can be chosen so that the uncertainty
about the shift after observing $F^\ell$ approaches
$\min\{H(Z|F),(1+c/\ell)\log_2d\}$ per channel use. Because an entanglement
rate $E$ allows $(c/\ell)\log_2d$ to approach $E$, this gives the achievable
commitment rate $\min\{H(Z|F),\log_2d+E\}$.

A separate difficulty is to prove security against arbitrary quantum
adversaries. Classical commitment results for the induced channels apply only
when the parties use  classical encoding and therefore do
not directly establish security in our quantum-input model. For hiding, we
show that Bob's entire quantum output can be generated from the induced
classical output by a fixed post-processing map. For binding, we establish that  every  quantum
commit and reveal strategy induces the same  statistics as a
dishonest strategy against the underlying classical commitment~ code.

For the converse, we apply the hiding condition to a receiver who stores all
raw channel outputs, which bounds the mutual information between the
committed string and Bob's  commit-phase view. Correctness,
 binding, and Fano's inequality show that the string can be
recovered from this view together with the channel environment. By the chain
rule, the committed length is then bounded by
the additional information about the string obtained when the environment
is added to Bob's view, which we prove to be at most $nH(Z|F)$.  Additionally, 
we formalize the intuition that each $d$-dimensional channel input
contributes at most $\log_2d$ bits to the commitment rate, while preshared
entanglement of rate $E$ contributes at most $E$ additional bits per channel
use, giving the additional rate bound
$\log_2d+E$.

 We also consider fully interactive commit phases, in which the parties may
exchange authenticated classical messages before and after every channel
use, and Alice may choose each new channel input using Bob's preceding
messages. While the noninteractive restriction is essential for the upper bound
$nH(Z|F)$ because, conditioned on $F^n$, all channel uses act on a fixed joint input
state, we show that  the  upper bound
$\log_2d+E$ remains valid under such interaction. For the channels covered by our main result, this bound matches
the noninteractive capacity whenever
$H(Z|F)\geq\log_2d+E$, showing that interaction does not help in this
regime. Additionally, for the quantum erasure channel, we show that interaction does not increase the capacity at any
entanglement rate.

\emph{Comparison with the most closely related work.}
\cite{winter2003commitment} derived the commitment capacity
 for nonredundant discrete memoryless channels and discussed an
extension to cq  channels. Subsequently, \cite{hayashi2023commitment}
  established the commitment capacity of nonredundant cq
channels for noninteractive protocols and for interactive protocols
satisfying an invertibility condition. In those models, the channels have a  classical
input alphabet:
the sender selects a classical symbol $x$, which the channel
maps to a quantum state $\rho_x$, thus, the inputs of any sender strategy are represented by a classical
sequence $x^n$. By contrast, the channels considered in this work have quantum
inputs, and a dishonest sender may transmit arbitrary states, including states entangled across channel uses and
with a private quantum memory, so that inputs need not correspond to any classical input
sequence $x^n$. The achievability and converse proofs of
\cite{winter2003commitment,hayashi2023commitment}, 
therefore do not cover arbitrary quantum-input
strategies, and cannot be applied to our model. 

The remainder of the paper is organized as follows. Section
\ref{sec:notation} introduces the notation and terminology, Section~\ref{sec:problem_statement} formalizes the  problem statement,
Section \ref{sec:main} provides our main results, whose proofs are deferred to  Sections \ref{sec:achievability} and
\ref{sec:converse}, 
Section \ref{sec:interactive_erasure} discusses extensions to interactive protocols, and
Section \ref{sec:conclusion} provides concluding remarks.

\section{Notation}
\label{sec:notation}
Let  $d=p^r$ be a  
prime power and let $\mathbb F_d$ be the finite field
with $d$ elements. Following \cite{ketkar2006nonbinary}, for the computational basis
$\{|j\rangle:j\in\mathbb F_d\}$ and $z\triangleq (a,b)\in\mathbb F_d^2$, define
$X(a)|j\rangle\triangleq|j+a\rangle$, $Z(b)|j\rangle\triangleq\chi(bj)|j\rangle$, $W_{z}\triangleq X(a)Z(b)$, 
where $\chi:\mathbb F_d\to\mathbb C, t\mapsto \exp\!\left(
\frac{2\pi\mathrm i}{p}
\sum_{j=0}^{r-1} t^{p^j}
\right).$
 Using the terminology  in \cite{wilde2013quantum}, for $a,b\in\mathbb F_d$, $X(a)$ and $Z(b)$ denote the
cyclic-shift and phase operators, respectively, and their products $X(a)Z(b)$
are called Heisenberg--Weyl operators. Moreover, a channel that randomly applies
these operators is called a Pauli qudit channel. 

For
$z=(z_1,\ldots,z_\ell)\in(\mathbb F_d^2)^\ell$, we write
$
W_z\triangleq W_{z_1}\otimes\cdots\otimes W_{z_\ell}.
$ For $u_i=(s_i,t_i)$ and $z_i=(a_i,b_i)$ in $\mathbb F_d^2$, define
$
[u_i,z_i]\triangleq t_i a_i-s_i b_i,
$
and for $u,z\in(\mathbb F_d^2)^\ell$, define
$[u,z]\triangleq\sum_{i=1}^{\ell}[u_i,z_i].$ 
For a subspace $L\subseteq(\mathbb F_d^2)^\ell$, define
$
L^\perp
\triangleq
\{u\in(\mathbb F_d^2)^\ell:[u,v]=0,\ \forall v\in L\}.$ 
Moreover, $L$ is called self-orthogonal if $L\subseteq L^\perp$.  For a subset $\mathcal S$ of an $\mathbb F_d$-vector space,
$\operatorname{span}_{\mathbb F_p}\mathcal S$ denotes the set of all finite
linear combinations of elements of $\mathcal S$ with coefficients in  $\mathbb F_p$. For an $\mathbb F_d$-linear subspace $L$, $\dim_{\mathbb F_d}L$ denotes
its vector-space dimension over $\mathbb F_d$, so that
$|L|=d^{\dim_{\mathbb F_d}L}$.

Finally, following \cite{winter2003commitment}, a finite classical channel
$W:\mathcal X\to\mathcal Y$ is nonredundant if, for every
$x\in\mathcal X$, $W(\cdot|x)
\notin
\operatorname{conv}
\left\{
W(\cdot|x'):
x'\neq x
\right\}.$

\section{Problem Statement}
\label{sec:problem_statement}

\subsection{Channel model}

Let $\mathcal F$ be a finite set, and let $q$ be a probability distribution
on $\mathcal F\times\mathbb F_d^2$. Define the channel
$\mathcal N_q:A\to FB$ by
\begin{equation}
\mathcal N_q(\rho)
=
\sum_{f\in\mathcal F}
\sum_{z\in\mathbb F_d^2}
q(f,z)
|f\rangle\!\langle f|^F
\otimes
W_z\rho W_z^\dagger.
\label{eq:flagged_weyl_channel}
\end{equation}
One isometric extension is
\begin{equation}
V_q|\psi\rangle^A
=
\sum_{f,z}\sqrt{q(f,z)}
|f\rangle^F
W_z|\psi\rangle^B
|f,z\rangle^E,
\label{eq:flagged_weyl_isometry}
\end{equation}
where the environment $E$ is inaccessible to Alice and Bob. 

\subsection{Commitment protocols}

As formalized below, we study commitment to a classical $k$-bit string over $n$ independent uses
of $\mathcal N_q$ when the reveal phase uses authenticated noiseless classical
communication.  

\begin{definition}
An $(n,k,K)$ noninteractive classical string commitment protocol over
$\mathcal N_q$ consists of:

\begin{itemize}
\item A committed message $S\in\{0,1\}^k$, held by Alice.

\item A preshared pure state $\Phi^{E_AE_B}$ of Schmidt rank at most $K$,
where Alice holds $E_A$ and Bob holds $E_B$.

\item A noninteractive commit phase. Alice prepares a joint state on
$A^nR_A$, possibly by acting jointly on $S$ and $E_A$, where $A^n$ denotes
the channel-input systems and $R_A$ denotes the private quantum register
that Alice retains. She sends $A^n$ through $\mathcal N_q^{\otimes n}$ and
receives no information from Bob during this phase. Bob may perform arbitrary operations on $F^nB^nE_B$ and
private ancillas. At the end of the commit phase, Alice holds $R_A$ and Bob holds a classical
register $C_{\mathrm{com}}$, containing all classical information that he
has generated, and a quantum register $B_{\mathrm{aux}}$, containing all
quantum information that he retains.  

\item A reveal phase. For each $s\in\{0,1\}^k$, Alice has an honest reveal
map
$\mathsf{Rev}_s:R_A\to D_{\mathrm{com}}R_A'$, where
$D_{\mathrm{com}}$ is a classical decommitment sent to Bob. Bob applies a
classical verification map
$\mathsf{Ver}:C_{\mathrm{com}}D_{\mathrm{com}}\to\widehat S J$, where
$\widehat S\in\{0,1\}^k$ and $J\in\{0,1\}$. If $J=0$, Bob accepts and
outputs $\widehat S$, whereas if $J=1$, Bob rejects.
\end{itemize}
\end{definition}

For an honest commitment to $s\in\{0,1\}^k$, let
$\Theta_s^{C_{\mathrm{com}}B_{\mathrm{aux}}R_AE^n}$ denote the global
state at the end of the commit phase. Bob's accessible post-commit state is
\begin{equation*}
\theta_s^{C_{\mathrm{com}}B_{\mathrm{aux}}}
\triangleq
\operatorname{Tr}_{R_AE^n}
\left[
\Theta_s^{C_{\mathrm{com}}B_{\mathrm{aux}}R_AE^n}
\right].
\end{equation*}
For $\hat s\in\{0,1\}^k$ and $j\in\{0,1\}$, the honest reveal and
verification outcomes satisfy
\begin{equation*}
\Prb[J=j,\widehat S=\hat s\mid S=s]
=
\operatorname{Tr}
\left[
\left(
|\hat s\rangle\!\langle\hat s|^{\widehat S}
\otimes
|j\rangle\!\langle j|^J
\right)
\Omega_s^{\widehat S J}
\right],
\end{equation*}
where $\Omega_s^{\widehat S J}$ is obtained by applying the honest reveal
map and verifier to $\Theta_s$ and tracing out all other systems.

\begin{definition}
Fix $n,k,K\in\mathbb N$ and $\alpha,\beta,\gamma\in[0,1]$. An $(n,k,K)$
protocol is an $(n,k,K,\alpha,\beta,\gamma)$ noninteractive classical string
commitment protocol if it satisfies the following conditions.

\medskip
\noindent\emph{1. Correctness.} For every $s\in\{0,1\}^k$, when Alice
honestly commits to $s$ and both parties are honest,
$$
\Prb[J=0,\widehat S=s\mid S=s]\geq1-\alpha.
$$

\medskip
\noindent\emph{2. Hiding.}
For every dishonest Bob commit strategy
$\mathsf B_{\mathrm{com}}^\star$, there exists a state
$\tau_{\mathsf B_{\mathrm{com}}^\star}^{C_{\mathrm{com}}B_{\mathrm{aux}}}$,
independent of $s$, such that for every $s\in\{0,1\}^k$,
\begin{equation*}
\left\|
\theta_{s,\mathsf B_{\mathrm{com}}^\star}^{C_{\mathrm{com}}B_{\mathrm{aux}}}
-
\tau_{\mathsf B_{\mathrm{com}}^\star}^{C_{\mathrm{com}}B_{\mathrm{aux}}}
\right\|_1
\leq\beta.
\end{equation*}

\medskip
\noindent\emph{3. Binding.}
Consider any dishonest Alice commit strategy while Bob follows his honest
commit strategy, and let
$\rho^{C_{\mathrm{com}}B_{\mathrm{aux}}E^nR_A}$ be the state at the end of
the commit phase. For this commit strategy and each possible value $c$ of the classical
register $C_{\mathrm{com}}$, there exists a POVM  $
\mathsf{Ext}_c
=
\left\{
\Lambda_{\widetilde s,c}^{B_{\mathrm{aux}}E^n}
\right\}_{\widetilde s\in\{0,1\}^k}$,
such that, for every dishonest reveal strategy
$\mathsf{Rev}^\star:R_A\to D_{\mathrm{com}}R_A'$, we have
\begin{equation}
\Prb\left[
J=0,\;\widehat S\neq\widetilde S
\right]
\leq\gamma,
\label{eq:extractable_binding}
\end{equation}
where the value $C_{\mathrm{com}}=c$ selects the POVM
$\mathsf{Ext}_c$, which is applied  to $B_{\mathrm{aux}}E^n$, and
$\widetilde S$ denotes its outcome.  

\end{definition}

\begin{remark}
\label{rem:binding_interpretation}
$\{\mathsf{Ext}_c\}_{c}$ may depend on Alice's dishonest commit strategy, but it
must be chosen independently of her subsequent reveal strategy and must
satisfy \eqref{eq:extractable_binding} for every such strategy. Thus, 
$\widetilde S$ is a reference value obtained from the
post-commit state by a measurement fixed before Alice chooses her reveal
strategy. The measurement need not be accessible to
either party, its role is solely to certify that Alice cannot choose among
different accepted values after the commit phase.
\end{remark}

\begin{remark} \label{rem2}
Suppose that the channel and all strategies are classical. For any dishonest
strategy that generates two decommitments, if Bob accepts both and the two
accepted values differ, then at least one differs from $\widetilde S$.
Consequently, \eqref{eq:extractable_binding} implies the pairwise binding
condition of \cite{winter2003commitment}, with binding error at most
$2\gamma$. 
\end{remark}

\begin{definition}
Fix $E\geq0$. A rate $R\geq0$ is noninteractively achievable with
entanglement rate $E$ if there exists a sequence of
$(n,k_n,K_n,\alpha_n,\beta_n,\gamma_n)$ protocols such that
\begin{equation*}
\liminf_{n\to\infty}\frac{k_n}{n}\geq R,
\qquad
\limsup_{n\to\infty}\frac{\log_2K_n}{n}\leq E,
\qquad \lim_{n\to\infty}
\alpha_n=\lim_{n\to\infty}
\beta_n=\lim_{n\to\infty}
\gamma_n=0.
\end{equation*}
The supremum of all such rates is denoted by
$C_{\mathrm{com}}(E;\mathcal N_q)$. The capacity with no preshared
entanglement, i.e., $K_n=1$, is denoted by $C_{\mathrm{com}}(\mathcal N_q)$, and the
capacity with unrestricted preshared entanglement is denoted by
$C_{\mathrm{com}}^{\mathrm{EA}}(\mathcal N_q)$.
\end{definition}
\section{Main Results}
\label{sec:main}

Our main result determines the noninteractive commitment capacity as a
function of the available entanglement~rate.

\begin{theorem}
\label{thm:flagged_weyl_capacity}
For each $f\in\mathcal F$, define the Fourier coefficients \cite{terras1999fourier}  of the function
$z\mapsto q(f,z)$ by
$\widehat q_f(u)
\triangleq
\sum_{z\in\mathbb F_d^2}
q(f,z)\chi([u,z])$, $u\in\mathbb F_d^2$, and define
$\mathcal S_q
\triangleq
\left\{
u\in\mathbb F_d^2:
\widehat q_f(u)\neq0
\text{ for some }f\in\mathcal F
\right\}$.

Let $\mathcal N_q$ be the channel in
\eqref{eq:flagged_weyl_channel}. If
$\mathcal S_q=\mathbb F_d^2$, then, for every $E\geq0$,
\begin{equation}
C_{\mathrm{com}}(E;\mathcal N_q)
=
\min\{H(Z|F),\log_2d+E\}.
\label{eq:flagged_weyl_entanglement_tradeoff}
\end{equation}
Consequently,
\begin{align}
C_{\mathrm{com}}^{\mathrm{EA}}(\mathcal N_q)
&=
H(Z|F),
\label{eq:flagged_weyl_ea_capacity}\\
C_{\mathrm{com}}(\mathcal N_q)
&=
\min\{H(Z|F),\log_2d\}. \notag
\end{align}
Moreover, \eqref{eq:flagged_weyl_ea_capacity} remains valid under the weaker
condition
$\operatorname{span}_{\mathbb F_p}\mathcal S_q=\mathbb F_d^2$.
\end{theorem}

\begin{proof}
The achievability is proved in Section \ref{sec:achievability}, and the
converse is proved in Section \ref{sec:converse}.
\end{proof}

\begin{remark}
    The condition $\mathcal S_q=\mathbb F_d^2$ guarantees that every
classical channel induced in our achievability construction is
nonredundant. In  the unrestricted entanglement-assisted case, the weaker condition
$\operatorname{span}_{\mathbb F_p}\mathcal S_q=\mathbb F_d^2$ is
sufficient.
\end{remark}

We first specialize Theorem \ref{thm:flagged_weyl_capacity} to  quantum erasure channels \cite{bennett1997capacities}.

\begin{corollary}
\label{cor:qudit_erasure_capacity}
For $0<\epsilon<1$, define the quantum erasure channel
\begin{equation*}
\mathcal E_{\epsilon,d}(\rho)
=
(1-\epsilon)\rho
+
\epsilon|e\rangle\!\langle e|.
\end{equation*}
For every $E\geq0$,
\begin{equation*}
C_{\mathrm{com}}(E;\mathcal E_{\epsilon,d})
=
\min\{2\epsilon\log_2d,\log_2d+E\}.
\end{equation*}
Consequently, $C_{\mathrm{com}}^{\mathrm{EA}}(\mathcal E_{\epsilon,d})
=
2\epsilon\log_2d$, and $C_{\mathrm{com}}(\mathcal E_{\epsilon,d})
=
\min\{2\epsilon\log_2d,\log_2d\}.$
\end{corollary}

\begin{proof}
We first use the following equivalent form of the  erasure channel,
as required by Theorem \ref{thm:flagged_weyl_capacity}: $\widetilde{\mathcal E}_{\epsilon,d}(\rho)
=
(1-\epsilon)|0\rangle\!\langle0|^F\otimes\rho
+
\epsilon|1\rangle\!\langle1|^F\otimes I/d$. The equivalence
follows because Bob can replace the erasure symbol $|e\rangle$ with
$|1\rangle\!\langle1|^F\otimes I/d$, or perform the reverse replacement.  Then, using the identity
$
d^{-2}\sum_{z\in\mathbb F_d^2}W_z\rho W_z^\dagger=I/d
$, e.g., \cite[Eq. (4.349)]{wilde2013quantum}, we see that $\widetilde{\mathcal E}_{\epsilon,d}$ is the channel
$\mathcal N_q$ associated with
$q(0,0)=1-\epsilon$ and
$q(1,z)=\epsilon/d^2$ for every $z\in\mathbb F_d^2$, with all other
probabilities equal to zero. Thus, conditioned on $F=0$, the applied
operator index is $Z=0$, whereas conditioned on $F=1$, it is uniform on
$\mathbb F_d^2$ so that $
H(Z|F)
=
(1-\epsilon)H(Z|F=0)+\epsilon H(Z|F=1)
=
2\epsilon\log_2d.
$ Finally, \eqref{eq:flagged_weyl_entanglement_tradeoff} holds because $\mathcal S_q=\mathbb F_d^2$ since for every $u\in\mathbb F_d^2$,
$
\widehat q_0(u)
=
\sum_{z\in\mathbb F_d^2}q(0,z)\chi([u,z])
=
1-\epsilon
\neq0.
$
\end{proof}

We next specialize Theorem \ref{thm:flagged_weyl_capacity} to Pauli qudit channels, which corresponds to the special case $F$ constant.

\begin{corollary}
\label{cor:qudit_weyl_subclasses}
Let $q$ be a distribution on $\mathbb F_d^2$ and define the Pauli qudit
channel
\begin{equation*}
\mathcal W_q(\rho)
=
\sum_{a,b\in\mathbb F_d}
q(a,b)W_{(a,b)}\rho W_{(a,b)}^\dagger.
\end{equation*}
If $\mathcal S_q=\mathbb F_d^2$, then, for every $E\geq0$,
\begin{equation*}
C_{\mathrm{com}}(E;\mathcal W_q)
=
\min\{H(q),\log_2d+E\}.
\end{equation*}

In particular:
\begin{enumerate}[(i)]

\item For $r$ a distribution on $\mathbb F_d$, define
$\widehat r(t)\triangleq\sum_{x\in\mathbb F_d}r(x)\chi(tx)$, and consider
the phase and cyclic-shift channels
\begin{align*}
\mathcal Z_r(\rho)
 =
\sum_{b\in\mathbb F_d}
r(b)Z(b)\rho Z(b)^\dagger,\qquad
\mathcal X_r(\rho)
=
\sum_{a\in\mathbb F_d}
r(a)X(a)\rho X(a)^\dagger.
\end{align*}
If $\widehat r(t)\neq0$ for every
$t\in\mathbb F_d\setminus\{0\}$, then, for every $E\geq0$,
\begin{equation*}
C_{\mathrm{com}}(E;\mathcal Z_r)
=
C_{\mathrm{com}}(E;\mathcal X_r)
=
H(r).
\end{equation*}

\item For $q_X$ and $q_Z$ distributions on $\mathbb F_d$, define
$\widehat q_X$ and $\widehat q_Z$ analogously to $\widehat r$, and consider
the Pauli qudit channel with independent shift and phase indices
\begin{equation*}
\mathcal W_{q_X,q_Z}(\rho)
=
\sum_{a,b\in\mathbb F_d}
q_X(a)q_Z(b)W_{(a,b)}\rho W_{(a,b)}^\dagger.
\end{equation*}
If $\widehat q_X(t)\neq0$ and $\widehat q_Z(t)\neq0$ for every
$t\in\mathbb F_d\setminus\{0\}$, then, for every $E\geq0$,
\begin{equation*}
C_{\mathrm{com}}(E;\mathcal W_{q_X,q_Z})
=
\min\{H(q_X)+H(q_Z),\log_2d+E\}.
\end{equation*}

\end{enumerate}
\end{corollary}

\begin{proof}
For the phase and cyclic-shift channels, the condition on
$\widehat r$ gives $\mathcal S_q=\mathbb F_d^2$, and
$H(r)\leq\log_2d$. For $\mathcal W_{q_X,q_Z}$, one can verify that for $u=(s,t)$,
$
\widehat q(s,t)
=
\widehat q_X(t)\widehat q_Z(-s)
$, so that  the  conditions on $\widehat q_X$ and $\widehat q_Z$ give
$\mathcal S_q=\mathbb F_d^2$. 
\end{proof}

We next consider the qubit specialization of Corollary \ref{cor:qudit_weyl_subclasses}.
\begin{corollary} 
\label{cor:qubit_pauli_capacity}
Let $q=(q_I,q_X,q_Y,q_Z)$ be a distribution and define the qubit Pauli channel
\begin{equation*}
\mathcal P_q(\rho)
=
q_I\rho+q_XX\rho X+q_YY\rho Y+q_ZZ\rho Z.
\end{equation*}
For the three nonzero elements of $\mathbb F_2^2$, the  Fourier coefficients of $q$ are
\begin{align}
\widehat q(1,0)
=
q_I+q_X-q_Y-q_Z, \quad
\widehat q(1,1)
=
q_I+q_Y-q_X-q_Z, \quad
\label{eq:qubit_pauli_nondegeneracy}
\widehat q(0,1)
=
q_I+q_Z-q_X-q_Y.
\end{align}
If all three coefficients
in \eqref{eq:qubit_pauli_nondegeneracy} are nonzero, then, for every
$E\geq0$,
\begin{equation*}
C_{\mathrm{com}}(E;\mathcal P_q)
=
\min\{H(q_I,q_X,q_Y,q_Z),1+E\}.
\end{equation*}
If at least two of the coefficients in
\eqref{eq:qubit_pauli_nondegeneracy} are nonzero, then
\begin{equation*}
C_{\mathrm{com}}^{\mathrm{EA}}(\mathcal P_q)
=
H(q_I,q_X,q_Y,q_Z).
\end{equation*}
\end{corollary}

\begin{proof}
The claims follow from Theorem
\ref{thm:flagged_weyl_capacity} with  $F$ constant because $\mathcal S_q=\mathbb F_2^2$ is equivalent to having all three coefficients
in \eqref{eq:qubit_pauli_nondegeneracy}  nonzero, and $\operatorname{span}_{\mathbb F_2}\mathcal S_q=\mathbb F_2^2$ is
equivalent to having at least two of the coefficients in
\eqref{eq:qubit_pauli_nondegeneracy}  nonzero.
\end{proof}

A  special case of Corollary \ref{cor:qubit_pauli_capacity} is the qubit depolarizing channel.

\begin{corollary} 
\label{cor:depolarizing_capacity}
For $0\leq p\leq1$, define the qubit depolarizing channel
\begin{equation}
\mathcal D_p(\rho)
=(1-p)\rho+p\frac{I}{2}.
\label{eq:depolarizing_channel}
\end{equation}
Define 
$q_p
=
\left(
1-\frac{3p}{4},
\frac p4,
\frac p4,
\frac p4
\right).$
For $0\leq p<1$ and every $E\geq0$,
\begin{equation*}
C_{\mathrm{com}}(E;\mathcal D_p)
=
\min\{H(q_p),1+E\}.
\end{equation*}
\end{corollary}

\begin{proof}
We write
$
\mathcal D_p(\rho)
=
\sum_{z\in\mathbb F_2^2}q_p(z)W_z\rho W_z^\dagger
$
and observe that the Fourier coefficient of $q_p$ satisfies
$
\widehat q_p(u)=1-p
$
at each nonzero index $u\in\mathbb F_2^2$.
Thus, for $p<1$, $\mathcal S_{q_p}=\mathbb F_2^2$, and the capacity follows
from Corollary \ref{cor:qubit_pauli_capacity}. 
\end{proof}

We next consider the dephrasure channel \cite{leditzky2018dephrasure}.
\begin{corollary} 
\label{cor:dephrasure}
For $0\leq\epsilon<1$ and $0\leq p\leq1$, define the dephrasure channel
\begin{equation*}
\mathcal N_{\epsilon,p}(\rho)
=
(1-\epsilon)\bigl((1-p)\rho+pZ\rho Z\bigr)
+\epsilon|e\rangle\!\langle e|,
\end{equation*}
where $|e\rangle$ is an erasure flag orthogonal to the qubit output space.
If $p\neq1/2$, then, for every $E\geq0$,
\begin{equation*}
C_{\mathrm{com}}(E;\mathcal N_{\epsilon,p})
=
\min\left\{
(1-\epsilon)h_2(p)+2\epsilon,\,
1+E
\right\}.
\end{equation*}
\end{corollary}

\begin{proof}
As in the proof of Corollary
\ref{cor:qudit_erasure_capacity}, express the dephrasure channel in the equivalent form 
$
\widetilde{\mathcal N}_{\epsilon,p}(\rho)
=
(1-\epsilon)|0\rangle\!\langle0|^F
\otimes\bigl((1-p)\rho+pZ\rho Z\bigr)
+
\epsilon|1\rangle\!\langle1|^F\otimes I/2,
$ with 
$
I/2
=
4^{-1}\sum_{z\in\mathbb F_2^2}W_z\rho W_z^\dagger,
$ so that 
$q(0,(0,0))=(1-\epsilon)(1-p)$,
$q(0,(0,1))=(1-\epsilon)p$, and
$q(1,z)=\epsilon/4$ for every $z\in\mathbb F_2^2$, with all other
probabilities equal to zero. Then, one can verify that $\widehat q_0(s,t)=1-\epsilon$ when $s=0$, and
$\widehat q_0(s,t)=(1-\epsilon)(1-2p)$ when $s=1$, which are nonzero for every
$(s,t)\in\mathbb F_2^2$ because 
$\epsilon<1$ and $p\neq1/2$. Hence,
$\mathcal S_q=\mathbb F_2^2$ and \eqref{eq:flagged_weyl_entanglement_tradeoff} holds. Finally, $
H(Z|F)
= (1-\epsilon)  H(Z|F=0) + \epsilon H(Z|F=1)=
(1-\epsilon)h_2(p)+2\epsilon.
$
\end{proof}

\section{Achievability of Theorem \ref{thm:flagged_weyl_capacity}}
\label{sec:achievability}

The achievability proof idea is to convert blocks of the quantum channel
into a classical channel and then apply a classical
commitment code. Section
\ref{subsec:classical_commitment_codes} reviews  this classical
commitment code.
Section \ref{subsec:induced_classical_channel} constructs the induced classical
channel
and shows that  the uniform-input equivocation of this channel approaches the
desired rate. Section~\ref{subsec:additive_channel_nonredundancy} determines a sufficient condition to ensure that the induced channel is nonredundant.
Section~\ref{subsec:common_security_reduction} then shows that the
 protocol is correct, hiding, and binding against arbitrary
quantum adversaries. Finally,
Section~\ref{subsec:completion_achievability} applies the
preceding results to prove the claimed rates.

\subsection{Classical commitment codes}
\label{subsec:classical_commitment_codes}

We first show that the result in \cite{winter2003commitment} extends to our binding requirement, see Remark \ref{rem2}.

\begin{theorem}
\label{thm:dmc_extractable_commitment_achievability}
Let $W:\mathcal X\to\mathcal Y$ be a finite nonredundant classical channel.
Then, every rate $R<\max_{P_X}H(X|Y)$
is achievable over $W$ in the absence of commit-phase public
communication.
\end{theorem}

\begin{proof}
See Appendix \ref{App_th}.
\end{proof}

\subsection{Induced classical channel}
\label{subsec:induced_classical_channel}

We show how $\ell$ uses of $\mathcal N_q$, together with $c$ maximally
entangled qudit pairs, simulate one use of a classical channel.
The construction uses a subspace $L$ to group Heisenberg--Weyl operator
indices that have the same effect on the encoded states.

We will use the following lemma, which is a consequence of \cite{ketkar2006nonbinary},  to construct the channel inputs.
\begin{lemma}[Adapted from \cite{ketkar2006nonbinary}] 
\label{lem:self_orthogonal_joint_eigenspace}
Let $c\in\{0,\ldots,\ell\}$, and let
$L\subseteq(\mathbb F_d^2)^\ell$ be an $\mathbb F_d$-linear subspace
satisfying
$
L\subseteq L^\perp
$
and
$
\dim_{\mathbb F_d}L=\ell-c.
$
There exist  $\{\lambda_v:v\in L\}$ and 
$\mathcal Q_L\subseteq(\mathbb C^d)^{\otimes\ell}$ such that
$W_v|\psi\rangle
=
\lambda_v|\psi\rangle$, 
$v\in L$, $|\psi\rangle\in\mathcal Q_L$,
and
$
\dim\mathcal Q_L=d^c.
$
Moreover, the orthogonal projector onto $\mathcal Q_L$ is
$
P_L
=
\frac{1}{|L|}
\sum_{v\in L}\lambda_v^{-1}W_v.
$
\end{lemma}
\begin{proof}
Lemma \ref{lem:self_orthogonal_joint_eigenspace} is a consequence of \cite[Lem.~9]{ketkar2006nonbinary} and the proof of \cite[Th.~13]{ketkar2006nonbinary}, we include its
proof in Appendix~\ref{App_lem1} for completeness.
\end{proof}

Consider the subspace $\mathcal Q_L$ and its projector $P_L$
provided by Lemma~\ref{lem:self_orthogonal_joint_eigenspace}. Choose an
orthonormal basis
$\{|\phi_j\rangle\}_{j=1}^{d^c}$ of $\mathcal Q_L$, and let
$R\simeq\mathbb C^{d^c}$ have orthonormal basis
$\{|j\rangle^R\}_{j=1}^{d^c}$. Define
$
|\Phi_L\rangle^{A^\ell R}
\triangleq
\frac{1}{\sqrt{d^c}}
\sum_{j=1}^{d^c}
|\phi_j\rangle^{A^\ell}|j\rangle^R,
$
which is maximally entangled between the subspace $\mathcal Q_L$ of
Alice's $\ell$ qudits and Bob's register $R$. We use this state to encode a classical input into the $\ell$ channel uses.
The classical input alphabet is
$
\mathcal X_L\triangleq(\mathbb F_d^2)^\ell/L,
$
the set of cosets of $L$, and we define
$
\pi_L:(\mathbb F_d^2)^\ell\to\mathcal X_L,
\pi_L(x)\mapsto x+L.
$

To send the classical symbol $\bar x=x+L\in\mathcal X_L$, Alice applies
$W_x$ to her system $A^\ell$, thereby preparing
$
|\Phi_{\bar x}^L\rangle^{A^\ell R}
\triangleq
(W_x\otimes I^R)|\Phi_L\rangle^{A^\ell R},
$
and sends $A^\ell$ through the $\ell$ channel uses, while Bob retains $R$.
Let
$
\Psi_{\bar x}^L
\triangleq
|\Phi_{\bar x}^L\rangle\!\langle\Phi_{\bar x}^L|.
$

We now determine the induced classical noise.
Let
$
(F^\ell,Z^\ell)
=
((F_1,\ldots,F_\ell),(Z_1,\ldots,Z_\ell))
$
have distribution $q^{\otimes\ell}$. Conditioned on
$F^\ell=f^\ell$ and $Z^\ell=z$, the channel reveals
$f^\ell$ to Bob and applies $W_z$ to the transmitted system so that, if
Alice's classical input is $\bar x=x+L$, then the joint state of the channel
output and Bob's register $R$ is
$
(W_z\otimes I^R)\Psi_{\bar x}^L
(W_z^\dagger\otimes I^R)
=
\Psi_{(x+z)+L}^L 
=
\Psi_{\bar x+\pi_L(z)}^L$,
where the first equality follows from the Heisenberg--Weyl multiplication
rule \cite[Eq.~(1)]{ketkar2006nonbinary}. We define the  joint distribution  $$q_L(f^\ell,\bar z)
\triangleq
\Prb[F^\ell=f^\ell,\pi_L(Z^\ell)=\bar z] 
=
\sum_{\substack{z:\,\pi_L(z)=\bar z}}
\prod_{i=1}^{\ell}q(f_i,z_i).$$

The next lemma shows that
$\{|\Phi_{\bar x}^L\rangle:\bar x\in\mathcal X_L\}$ is an orthonormal
basis so that Bob can measure the channel output together with $R$
in this basis. Moreover, if $\overline X$ is Alice's classical input and
$\overline Y$ is Bob's measurement outcome, the resulting classical
channel is
$
\overline Y
=
\overline X+\pi_L(Z^\ell),
$
with Bob also observing $F^\ell$.

\begin{lemma}
\label{lem:stabilizer_code_flagged_weyl_reduction}
$\{|\Phi_{\bar x}^L\rangle:\bar x\in\mathcal X_L\}$ is an orthonormal
basis of
$(\mathbb C^d)^{\otimes\ell}\otimes R$.

Suppose that Alice selects $\bar x=x+L$, applies $W_x$ to her share of
$|\Phi_L\rangle$, and sends her $\ell$ qudits through
$\mathcal N_q^{\otimes\ell}$. If Bob measures the channel output and $R$
in the basis $\{|\Phi_{\bar y}^L\rangle\}_{\bar y\in\mathcal X_L}$, then
his measurement outcome $\overline Y$ satisfies
$
\overline Y
=
\overline X+\pi_L(Z^\ell).
$
Together with the channel output $F^\ell$, this gives the classical
transition probabilities
\begin{equation}
W_{L,q}(f^\ell,\bar y|\bar x)
=
q_L(f^\ell,\bar y-\bar x).
\label{eq:stabilizer_quotient_additive_channel}
\end{equation}
If $\overline X$ is uniform on $\mathcal X_L$ and independent of
$(F^\ell,Z^\ell)$, then
\begin{equation*}
H(\overline X|F^\ell,\overline Y)
=
H\left(\pi_L(Z^\ell)\middle|F^\ell\right).
\end{equation*}
\end{lemma}
\begin{proof}
We first prove that the vectors $\{|\Phi_{\bar x}^L\rangle:\bar x\in\mathcal X_L\}$ are orthonormal. For each
$\bar x\in\mathcal X_L$, fix one representative $x\in(\mathbb F_d^2)^\ell$.
Since $W_x$ is unitary, $|\Phi_{\bar x}^L\rangle$ is a unit vector. Moreover, if $\bar x=x+L$ and $\bar y=y+L$ are distinct, then 
\begin{align*}
\langle\Phi_{\bar x}^L|\Phi_{\bar y}^L\rangle
&=
\langle\Phi_L|
(W_x^\dagger W_y\otimes I^R)
|\Phi_L\rangle
\\
&=
\operatorname{Tr} \left(
(W_x^\dagger W_y\operatorname{Tr}_R \left(
|\Phi_L\rangle \langle\Phi_L| \right) \right)
\\
&\overset{(a)}{=}
\frac{1}{d^c}
\operatorname{Tr}
\left(
P_LW_x^\dagger W_y
\right)
\\
&=
\frac{1}{d^c|L|}
\sum_{v\in L}
\lambda_v^{-1}
\operatorname{Tr}
\left(
W_vW_x^\dagger W_y
\right)
\\
&\overset{(b)}{=}
\frac{1}{d^c|L|}
\sum_{v\in L}
\alpha(v,x,y)
\operatorname{Tr}
\left(
W_{v+y-x}
\right)
\\
&\overset{(c)}{=}
0,
\end{align*}
where $(a)$ holds  because 
$\operatorname{Tr}_R \left(
|\Phi_L\rangle \langle\Phi_L| \right) = 
P_L/d^c$ with $P_L
=
|L|^{-1}\sum_{v\in L}\lambda_v^{-1}W_v
$ from Lemma~\ref{lem:self_orthogonal_joint_eigenspace},  $(b)$ follows from  
\cite[Eq.~(1)]{ketkar2006nonbinary} for some $\alpha(v,x,y)$, $(c)$ holds because, since $y-x\notin L$, we
have $v+y-x\neq0$ for every $v\in L$.   Finally, the number of these orthonormal vectors satisfies $ |\mathcal X_L|
=
\frac{|(\mathbb F_d^2)^\ell|}{|L|}
=
\frac{d^{2\ell}}{d^{\ell-c}}
=
d^{\ell+c}
=
d^\ell\dim R
=
\dim\left(
(\mathbb C^d)^{\otimes\ell}\otimes R
\right).$

We next determine the classical channel induced by measuring in this
basis. Conditioned on $Z^\ell=z$, the physical channel applies $W_z$,
therefore,
\begin{align*}
(W_z\otimes I^R)|\Phi_{\bar x}^L\rangle
&=
(W_zW_x\otimes I^R)|\Phi_L\rangle
\\
&=
\mathrm e^{\mathrm i\vartheta(x,z)}
(W_{x+z}\otimes I^R)|\Phi_L\rangle
\\
&=
\mathrm e^{\mathrm i\vartheta(x,z)}
|\Phi_{(x+z)+L}^L\rangle
\\
&=
\mathrm e^{\mathrm i\vartheta(x,z)}
|\Phi_{\bar x+\pi_L(z)}^L\rangle,
\end{align*}
where the second equality holds by  
\cite[Eq.~(1)]{ketkar2006nonbinary} with
$\mathrm e^{\mathrm i\vartheta(x,z)}$  an irrelevant  phase.
Consequently, if Alice inputs $\overline X=\bar x$ and Bob measures the
channel output and $R$ in the   basis $\{|\Phi_{\bar y}^L\rangle\}_{\bar y\in\mathcal X_L}$, then
$
\overline Y=\bar x+\pi_L(z)
$
conditioned on $Z^\ell=z$. Hence,
\begin{align}
\Prb[
F^\ell=f^\ell,\overline Y=\bar y
\mid
\overline X=\bar x
]
& = \sum_{z\in(\mathbb F_d^2)^\ell}
\Prb\!\left[
F^\ell=f^\ell,Z^\ell=z
\,\middle|\,
\overline X=\bar x
\right]
\Prb\!\left[
\overline Y=\bar y
\,\middle|\,
\overline X=\bar x,Z^\ell=z,F^\ell=f^\ell
\right] \notag \\
&=
\sum_{z\in(\mathbb F_d^2)^\ell}
\left(
\prod_{i=1}^{\ell}q(f_i,z_i)
\right)
\mathds 1
\left\{
\bar y=\bar x+\pi_L(z)
\right\}
\notag\\
&=
\sum_{\substack{z\in(\mathbb F_d^2)^\ell:\\
\pi_L(z)=\bar y-\bar x}}
\prod_{i=1}^{\ell}q(f_i,z_i)
\notag\\
&=
q_L(f^\ell,\bar y-\bar x),\notag
\end{align}
where the second equality holds because $(F^\ell,Z^\ell)$ is independent  of the input $\overline X$ and the  $\ell$ channel uses are independent.

Finally, we have
\begin{align}
H(\overline X\mid F^\ell,\overline Y)
&\overset{(a)}{=}
H(\pi_L(Z^\ell) \mid F^\ell,\overline Y)
\notag\\
&\overset{(b)}{=}
H(\pi_L(Z^\ell)\mid F^\ell), \notag
\end{align}
where $(a)$ holds because 
$
\overline X=\overline Y-\pi_L(Z^\ell)
$, $(b)$ holds because for every $(f^\ell,\bar y,\bar n)$ with
$\Prb[F^\ell=f^\ell]>0$,
\begin{align*}
\Prb[
\overline Y=\bar y,\pi_L(Z^\ell)=\bar n
\mid
F^\ell=f^\ell
]
&=
\Prb[
\pi_L(Z^\ell)=\bar n
\mid
F^\ell=f^\ell
]
\Prb[
\overline X=\bar y-\bar n
]
\\
&=
\frac{1}{|\mathcal X_L|}
\Prb[
\pi_L(Z^\ell)=\bar n
\mid
F^\ell=f^\ell
]
\\
&=
\Prb[
\overline Y=\bar y
\mid
F^\ell=f^\ell
]
\Prb[
\pi_L(Z^\ell)=\bar n
\mid
F^\ell=f^\ell
],
\end{align*}
where the first equality  holds because $\overline X$  is independent of
$(F^\ell,\pi_L(Z^\ell))$, and the second equality holds because $\overline X$ is uniform over~$\mathcal X_L$.
\end{proof}

\begin{remark}
When $c=0$, $\dim\mathcal Q_L=1$ and $R\simeq\mathbb C$, hence, the encoded states
$
\{|\Phi_{\bar x}^L\rangle:\bar x\in\mathcal X_L\}
$
are states of Alice's $\ell$ qudits only, and the construction uses no
preshared entanglement. When $c=\ell$, we have
$
L=\{0\},
$
$
\mathcal Q_L=(\mathbb C^d)^{\otimes\ell},
$
and
$
R\simeq(\mathbb C^d)^{\otimes\ell},
$
thus, choosing matching computational product bases for $A^\ell$ and
$R=R_1\cdots R_\ell$, the state $|\Phi_L\rangle^{A^\ell R}$ is the tensor
product of $\ell$ maximally entangled qudit pairs. With this choice,
$
\mathcal X_L=(\mathbb F_d^2)^\ell,
$
and the encoded family
$
\{|\Phi_x^L\rangle=(W_x\otimes I^R)|\Phi_L\rangle:
x\in(\mathbb F_d^2)^\ell\}
$
is the $\ell$-fold tensor product of the qudit Bell basis, e.g.,~\cite[Eq.~(3.234)]{wilde2013quantum}.
\end{remark}

The following lemma
  shows that  the uniform-input equivocation of  the induced classical
channel approaches the
desired rate.

\begin{lemma}
\label{lem:conditional_random_stabilizer_hashing}
Let $(F^\ell,Z^\ell)\sim q^{\otimes\ell}$, where
$Z_i\in\mathbb F_d^2$. Let
$c_\ell\in\{0,\ldots,\ell\}$ satisfy
$\lim_{\ell\to\infty}\frac{c_\ell}{\ell}=e$, $
0\leq e\leq1$.
There exists a sequence of self-orthogonal subspaces
$L_\ell\subseteq\mathbb F_d^{2\ell}$ with
$\dim L_\ell=\ell-c_\ell$ such that
\begin{equation*}
\lim_{\ell\to\infty}
\frac{1}{\ell}
H\left(
\pi_{L_\ell}(Z^\ell)
\middle|
F^\ell
\right)
=
\min\{H(Z|F),(1+e)\log_2d\}.
\end{equation*}
\end{lemma}
\begin{proof}Let $\mathcal L_\ell$ be the set of self-orthogonal subspaces of
$(\mathbb F_d^2)^\ell$ having dimension $\ell-c_\ell$, and choose $L$
uniformly from $\mathcal L_\ell$. Set
$
a_\ell\triangleq\ell^{2/3}, 
$
then, for every $f^\ell$ with positive probability, define
$$
\mathcal T(f^\ell)
\triangleq
\{z^\ell:
-\log_2P_{Z^\ell|F^\ell}(z^\ell|f^\ell)
\geq
\textstyle\sum_{i=1}^{\ell}H(Z|F=f_i)-a_\ell\},
$$
$
\tau(f^\ell)
\triangleq
\Prb[Z^\ell\in\mathcal T(f^\ell)|F^\ell=f^\ell],
$
$
V_{\max}
\triangleq
\max_{f:P_F(f)>0}
\operatorname{Var}_{Z\sim P_{Z|F}(\cdot|f)}
[-\log_2P_{Z|F}(Z|f)],
$
and
$
\kappa_\ell
\triangleq
1-\ell V_{\max}/a_\ell^2,
$
and
$$
T_\ell
\triangleq
\mathds{1}\{\textstyle\sum_{i=1}^{\ell}H(Z|F=F_i)
\geq\ell H(Z|F)-a_\ell\} \times 
\mathds{1}\{Z^\ell\in\mathcal T(F^\ell)\}.
$$
For all sufficiently large $\ell$, we have
\begin{align}
\Prb[T_\ell=1]
&=
\sum_{\substack{f^\ell:\\
\sum_{i=1}^{\ell}H(Z|F=f_i)
\geq
\ell H(Z|F)-a_\ell}}
\!\!\!\!\!\!\!\!\!\!\!\!\!\!\!\!\!\!\!\!P_{F^\ell}(f^\ell)
\Prb[
Z^\ell\in\mathcal T(f^\ell)
\mid
F^\ell=f^\ell
]
\notag\\
&=
\sum_{\substack{f^\ell:\\
\sum_{i=1}^{\ell}H(Z|F=f_i)
\geq
\ell H(Z|F)-a_\ell}}
\!\!\!\!\!\!\!\!\!\!\!\!\!\!\!\!\!\!\!\!P_{F^\ell}(f^\ell)
\Prb\left[
-\log_2P_{Z^\ell|F^\ell}(Z^\ell|f^\ell)
\geq
\textstyle\sum_{i=1}^{\ell}H(Z|F=f_i)-a_\ell
\,\middle|\,
F^\ell=f^\ell
\right]
\notag\\
&\overset{(a)}{\geq}
\kappa_\ell
\Prb\left[
\textstyle\sum_{i=1}^{\ell}H(Z|F=F_i)
\geq
\ell H(Z|F)-a_\ell
\right]
\notag\\
&=
\kappa_\ell
\left(
1-
\Prb\left[
\textstyle\sum_{i=1}^{\ell}H(Z|F=F_i)
<
\ell H(Z|F)-a_\ell
\right]
\right)
\notag\\
&\geq
\kappa_\ell
\left(
1-
\Prb\left[
\left|
\textstyle\sum_{i=1}^{\ell}H(Z|F=F_i)
-
\ell H(Z|F)
\right|
>
a_\ell
\right]
\right)
\notag\\
&\overset{(b)}{\geq}
\kappa_\ell
\left(
1-
\frac{
\ell\operatorname{Var}_{f\sim P_F}
[H(Z|F=f)]
}{
a_\ell^2
}
\right)
\notag\\
&\xrightarrow{\ell\to\infty}
1,
\label{eq:global_conditional_smoothing_probability}
\end{align}
where $(a)$ holds by Chebyshev's inequality because memorylessness gives
$-\log_2P_{Z^\ell|F^\ell}(Z^\ell|f^\ell)
=\sum_{i=1}^{\ell}-\log_2P_{Z|F}(Z_i|f_i)$, where the terms are independent, with respective means $H(Z|F=f_i)$ and
variances at most $V_{\max}$, $(b)$ holds by Chebyshev's
inequality because the random variables $H(Z|F=F_i)$ are independent and
identically distributed with mean
$
\sum_fP_F(f)H(Z|F=f)=H(Z|F).
$

Conditioned on $F^\ell=f^\ell$ and $T_\ell=1$, let
$\widetilde Z_{f^\ell}$ have the conditional distribution of $Z^\ell$,
that is,
$\Prb[\widetilde Z_{f^\ell}=z^\ell]
=\Prb[Z^\ell=z^\ell\mid F^\ell=f^\ell,T_\ell=1]$,
and let $\widetilde Z_{f^\ell}'$ be an independent copy. Then, we have
\begin{align}
\E_L\left[
\sum_{x\in\mathcal X_L}
\Prb\left[
\pi_L(\widetilde Z_{f^\ell})=x
\right]^2
\right]
&\overset{(a)}{=}
\Prb\left[
\pi_L(\widetilde Z_{f^\ell})
=
\pi_L(\widetilde Z_{f^\ell}')
\right]
\notag\\
&=
\Prb\left[
\widetilde Z_{f^\ell}
=
\widetilde Z_{f^\ell}'
\right]
+
\Prb\left[
\widetilde Z_{f^\ell}
\neq
\widetilde Z_{f^\ell}',
\,
\pi_L(\widetilde Z_{f^\ell})
=
\pi_L(\widetilde Z_{f^\ell}')
\right]
\notag\\
&\overset{(b)}{=}
\sum_{z^\ell}
\Prb\left[
\widetilde Z_{f^\ell}=z^\ell
\right]^2
+
\frac{d^{\ell-c_\ell}-1}{d^{2\ell}-1}
\left(
1-
\sum_{z^\ell}
\Prb\left[
\widetilde Z_{f^\ell}=z^\ell
\right]^2
\right)
\notag\\
&\leq
\sum_{z^\ell}
\Prb\left[
\widetilde Z_{f^\ell}=z^\ell
\right]^2
+
\frac{d^{\ell-c_\ell}-1}{d^{2\ell}-1}
\notag\\
&\overset{(c)}{<}
\sum_{z^\ell}
\Prb\left[
\widetilde Z_{f^\ell}=z^\ell
\right]^2
+
d^{-(\ell+c_\ell)}
\notag\\
&\leq
\max_{z^\ell}
\Prb\left[
\widetilde Z_{f^\ell}=z^\ell
\right]
+
d^{-(\ell+c_\ell)}
\notag\\
&=
\frac{1}{\tau(f^\ell)}
\max_{z^\ell\in\mathcal T(f^\ell)}
P_{Z^\ell|F^\ell}(z^\ell|f^\ell)
+
d^{-(\ell+c_\ell)}
\notag\\
&\overset{(d)}{\leq}
\kappa_\ell^{-1}
2^{-\sum_{i=1}^{\ell}H(Z|F=f_i)+a_\ell}
+
d^{-(\ell+c_\ell)}
\notag\\
&\overset{(e)}{\leq}
\kappa_\ell^{-1}
2^{-\ell H(Z|F)+2a_\ell}
+
d^{-(\ell+c_\ell)},
\label{eq:conditional_average_quotient_collision}
\end{align}
where $(a)$ holds because, for each fixed $L$, the random variables
$\pi_L(\widetilde Z_{f^\ell})$ and
$\pi_L(\widetilde Z_{f^\ell}')$ are independent and identically distributed, $(b)$ holds because for distinct
$z,z'\in\mathbb F_d^{2\ell}$, $\Prb_L[\pi_L(z)=\pi_L(z')]
=
\Prb_L[z-z'\in L]
=
\frac{d^{\ell-c_\ell}-1}{d^{2\ell}-1}$, which follows from
\begin{align}
(d^{2\ell}-1)\Prb_L[z-z'\in L]
&=
\sum_{v\neq0}\Prb[v\in L]
\notag\\
&=
\E_L\left[
\sum_{v\neq0}\mathds{1}\{v\in L\}
\right]
\notag\\
&=
\E_L\left[|L\setminus\{0\}|\right]
\notag\\ \notag
&=
d^{\ell-c_\ell}-1,
\end{align}
where the first equality holds because $(\mathbb F_d^2)^\ell$ contains $d^{2\ell}-1$ nonzero
vectors and, by~\cite[Th.~8]{matsumoto2002lower}, for every pair of nonzero
vectors $v,w\in(\mathbb F_d^2)^\ell$, there exists an invertible linear
map $S$ such that
$Sv=w$
and  $L\mapsto SL$ is a bijection of
$\mathcal L_\ell$, so that  $SL$ is uniform on
$\mathcal L_\ell$ and $\Prb[v\in L]
=
\Prb[Sv\in SL]
=
\Prb[w\in SL]
 =
\Prb[w\in L]$, the fourth equality holds because   $L\in\mathcal L_\ell$ contains
$d^{\ell-c_\ell}-1$ nonzero vectors, $(c)$ holds because $\frac{d^{\ell-c_\ell}-1}{d^{2\ell}-1}<
d^{\ell-c_\ell}/d^{2\ell}=d^{-(\ell+c_\ell)}$, $(d)$ holds because Chebyshev's inequality gives
$\tau(f^\ell)\geq\kappa_\ell$, $(e)$ holds because $T_\ell=1$
implies
$
\sum_{i=1}^{\ell}H(Z|F=f_i)
\geq
\ell H(Z|F)-a_\ell.
$

Averaging
\eqref{eq:conditional_average_quotient_collision} over
$F^\ell$ conditioned on $T_\ell=1$ gives
\begin{align}
\E_L\Bigg[
\sum_{f^\ell}
P_{F^\ell|T_\ell}(f^\ell|1)
\sum_x
\Prb\left[
\pi_L(\widetilde Z_{f^\ell})=x
\right]^2
\Bigg]
&\leq
\kappa_\ell^{-1}
2^{-\ell H(Z|F)+2a_\ell}
+
d^{-(\ell+c_\ell)}.
\label{eq:averaged_selected_stabilizer}
\end{align}
Therefore, there exists
$L_\ell\in\mathcal L_\ell$ for which the same averaged bound holds. For this choice, we have
\begin{align}
\frac{1}{\ell}H\left(
\pi_{L_\ell}(Z^\ell)
\middle|
F^\ell
\right)
&\geq
\frac{\Prb[T_\ell=1]}{\ell}
H\left(
\pi_{L_\ell}(Z^\ell)
\middle|
F^\ell,T_\ell=1
\right)
\notag\\
&=
\frac{\Prb[T_\ell=1]}{\ell}
\sum_{f^\ell}
P_{F^\ell|T_\ell}(f^\ell|1)
H\left(
\pi_{L_\ell}(\widetilde Z_{f^\ell})
\right)
\notag\\
&\overset{(a)}{\geq}
-\frac{\Prb[T_\ell=1]}{\ell}
\sum_{f^\ell}
P_{F^\ell|T_\ell}(f^\ell|1)
\log_2\left(
\sum_{x\in(\mathbb F_d^2)^\ell/L_\ell}
\Prb\left[
\pi_{L_\ell}(\widetilde Z_{f^\ell})=x
\right]^2
\right)
\notag\\
&\overset{(b)}{\geq}
-\frac{\Prb[T_\ell=1]}{\ell}
\log_2\left[
\sum_{f^\ell}
P_{F^\ell|T_\ell}(f^\ell|1)
\sum_{x\in(\mathbb F_d^2)^\ell/L_\ell}
\Prb\left[
\pi_{L_\ell}(\widetilde Z_{f^\ell})=x
\right]^2
\right]
\notag\\
&\overset{(c)}{\geq}
-\frac{\Prb[T_\ell=1]}{\ell}
\log_2\left[
\kappa_\ell^{-1}
2^{-\ell H(Z|F)+2a_\ell}
+
d^{-(\ell+c_\ell)}
\right]
\notag\\
&\overset{(d)}{\geq}
\frac{\Prb[T_\ell=1]}{\ell}
\min\left\{
\ell H(Z|F)-2a_\ell+\log_2\kappa_\ell,
(\ell+c_\ell)\log_2d
\right\}
-\frac{1}{\ell}
\notag\\
&\xrightarrow{\ell\to\infty}
\min\left\{
H(Z|F),(1+e)\log_2d
\right\}, \notag
\end{align}
where $(a)$ holds by \cite[Lemma~2.10.1]{cover1991information}, $(b)$ holds by Jensen's inequality,  $(c)$
holds by  
\eqref{eq:averaged_selected_stabilizer},  $(d)$ holds
because
$
-\log_2(u+v)
\geq
\min\{-\log_2u,-\log_2v\}-1
$
for $u,v>0$, the limit  holds
because
$
\Prb[T_\ell=1]\to1
$ by \eqref{eq:global_conditional_smoothing_probability},  
$
a_\ell=o(\ell),
$
$
\kappa_\ell\to1,
$
and
$
c_\ell/\ell\to e.
$
On the other hand,
\begin{align}
\limsup_{\ell\to\infty}
\frac{1}{\ell}
H\left(
\pi_{L_\ell}(Z^\ell)
\middle|
F^\ell
\right)
&\overset{(a)}{\leq}
\limsup_{\ell\to\infty}
\frac{1}{\ell}
\min\left\{
\log_2|\mathcal X_{L_\ell}|,
H(Z^\ell|F^\ell)
\right\}
\notag\\
&\overset{(b)}{=}
\limsup_{\ell\to\infty}
\min\left\{
\left(1+\frac{c_\ell}{\ell}\right)\log_2d,
H(Z|F)
\right\}
\notag\\
&=
\min\left\{
(1+e)\log_2d,
H(Z|F)
\right\}, \notag
\end{align}
where $(a)$ holds because $\pi_{L_\ell}(Z^\ell)$ is a deterministic
function of $Z^\ell$ taking values in $\mathcal X_{L_\ell}$, $(b)$
holds because
$
|\mathcal X_{L_\ell}|=d^{\ell+c_\ell}
$
and $(F^\ell,Z^\ell)$ is memoryless. 
\end{proof}

\subsection{Nonredundancy of the induced channel}
\label{subsec:additive_channel_nonredundancy}
 We now give a
condition under which the induced channel $W_{L,q}$ is nonredundant. With $\mathcal S_q$ defined in Theorem~\ref{thm:flagged_weyl_capacity}, define $\mathcal S_{q,L}
\triangleq
\left\{
u=(u_1,\ldots,u_\ell)\in L^\perp:
u_i\in\mathcal S_q
\text{ for every }i\in[\ell]
\right\}.$
\begin{lemma}
\label{lem:stabilizer_induced_channel_nondegeneracy}
The induced channel $W_{L,q}$ of
Lemma~\ref{lem:stabilizer_code_flagged_weyl_reduction} is nonredundant if
and only if
\begin{equation}
\operatorname{span}_{\mathbb F_p}\mathcal S_{q,L}
=
L^\perp.
\label{eq:quotient_fourier_nondegeneracy}
\end{equation}
In particular, $\mathcal S_q=\mathbb F_d^2$ guarantees nonredundancy for
every $L$. When $\ell=c=1$, \eqref{eq:quotient_fourier_nondegeneracy} reduces to $\operatorname{span}_{\mathbb F_p}\mathcal S_q
=
\mathbb F_d^2.$
\end{lemma}

\begin{proof}
 For every $u\in L^\perp$, one can verify that the Fourier coefficient of the
joint distribution $q_L$ of $F^\ell$ and $\pi_L(Z^\ell)$~is
\begin{align}
\sum_{\bar z\in\mathcal X_L}
q_L(f^\ell,\bar z)\chi([u,z])
&=
\prod_{i=1}^{\ell}\widehat q_{f_i}(u_i),
\label{eq:quotient_noise_fourier_coefficient}
\end{align}
so that $u\in\mathcal S_{q,L}$ if and only if $\textstyle\prod_{i=1}^{\ell}\widehat q_{f_i}(u_i)\neq0
\text{ for some }f^\ell$.

Suppose first that
$\operatorname{span}_{\mathbb F_p}\mathcal S_{q,L}\neq L^\perp$.
We prove that $W_{L,q}$ is redundant by finding $h\notin L$ such that
$q_L(f^\ell,\bar z+h+L)=q_L(f^\ell,\bar z)$ for every
$(f^\ell,\bar z)$, which  would imply that the distinct inputs $\bar x$
and $\bar x+h+L$ induce the same output distribution. To this end, it is sufficient to show that
$q_L(f^\ell,\bar z+h+L)$ and $q_L(f^\ell,\bar z)$ have the same Fourier
coefficients. For any $h\in(\mathbb F_d^2)^\ell$ and $u\in L^\perp$, the
Fourier coefficient of the translated distribution is
$\chi(-[u,h])\prod_{i=1}^{\ell}\widehat q_{f_i}(u_i)$. Therefore, it is
sufficient to find $h\notin L$ satisfying
$\chi([u,h])=1$ for every $u\in\mathcal S_{q,L}$, because 
 this condition implies $\chi(-[u,h])=1$ when
$u\in\mathcal S_{q,L}$, while the Fourier coefficient is zero for every
$f^\ell$ when $u\notin\mathcal S_{q,L}$.

We now show that such an $h$ exists. Let
$u_1,\ldots,u_r$ be an $\mathbb F_p$-basis of
$\operatorname{span}_{\mathbb F_p}\mathcal S_{q,L}$, so that
$r=\dim_{\mathbb F_p}\operatorname{span}_{\mathbb F_p}\mathcal S_{q,L}$.
Also, choose an $\mathbb F_p$-basis
$e_1+L,\ldots,e_N+L$ of $(\mathbb F_d^2)^\ell/L$, where
$N=\dim_{\mathbb F_p}((\mathbb F_d^2)^\ell/L)
=\dim_{\mathbb F_p}L^\perp$. Write the desired coset as
$h+L=\sum_{j=1}^{N}x_j(e_j+L)$, where $x_j\in\mathbb F_p$. For every
$k\in[r]$ and $j\in[N]$, let $c_{k,j}\in\mathbb F_p$ be defined by
$\chi([u_k,e_j])=\mathrm e^{2\pi\mathrm i c_{k,j}/p}$. Then,
$\chi([u_k,h])=1$ is equivalent to
$\sum_{j=1}^{N}c_{k,j}x_j=0$ in $\mathbb F_p$, which  gives $r$ homogeneous linear equations in the $N$ unknowns
$x_1,\ldots,x_N$. Since, by $\operatorname{span}_{\mathbb F_p}\mathcal S_{q,L}\neq L^\perp$,  $r<N$, the rank--nullity theorem guarantees a
nonzero solution, which defines a nonzero coset $h+L\neq L$. Finally, every $u\in\mathcal S_{q,L}$ is an $\mathbb F_p$-linear
combination of $u_1,\ldots,u_r$, hence, 
$\chi([u_k,h])=1$ for all $k\in[r]$ imply
$\chi([u,h])=1$ for every $u\in\mathcal S_{q,L}$.

Suppose now that
\eqref{eq:quotient_fourier_nondegeneracy} holds and, toward a
contradiction, assume that $W_{L,q}$ is redundant. Since
$W_{L,q}(f^\ell,\bar y\mid\bar x)
=
q_L(f^\ell,\bar y-\bar x)$, there exist $\lambda_{\bar h}\geq0$, indexed by
$\bar h\in\mathcal X_L\setminus\{L\}$, such that $ \sum_{\bar h\neq L}\lambda_{\bar h}=1$ and for every $(f^\ell,\bar y)$ 
\begin{equation}
q_L(f^\ell,\bar y)
=
\sum_{\bar h\neq L}
\lambda_{\bar h}
q_L(f^\ell,\bar y-\bar h)
\label{eq:redundant_quotient_row}.
\end{equation}
For any $u\in\mathcal S_{q,L}$, choosing $f^\ell$ for which the coefficient
in \eqref{eq:quotient_noise_fourier_coefficient} is nonzero, and taking the
Fourier transform of~\eqref{eq:redundant_quotient_row}  yields $1
=
\sum_{\bar h\neq L}
\lambda_{\bar h}\chi([u,h])$, which expresses $1$ as a
convex combination of complex numbers of modulus one. This is possible
only if every term with positive weight equals one, so that  if
$\lambda_{\bar h}>0$, then $\chi([u,h])=1$ for every
$u\in\mathcal S_{q,L}$. Under
\eqref{eq:quotient_fourier_nondegeneracy}, every $u\in L^\perp$ is an
$\mathbb F_p$-linear combination of vectors in $\mathcal S_{q,L}$, so
$\chi([u,h])=1$ for every $u\in L^\perp$ and, therefore,
$1=\chi([\alpha u,h])=\chi(\alpha[u,h])$ for every
$\alpha\in\mathbb F_d$. If
$[u,h]\neq0$, then $\alpha[u,h]$ ranges over all of $\mathbb F_d$ as $\alpha$
ranges over $\mathbb F_d$, which would imply that $\chi$ is identically
one, which contradicts the definition of $\chi$. Hence, $[u,h]=0$ for every $u\in L^\perp$, so
$h\in(L^\perp)^\perp=L$ so that $\bar h=L$, which contradicts
$\bar h\neq L$. Therefore, $W_{L,q}$ is~nonredundant.

Finally, if $\mathcal S_q=\mathbb F_d^2$, then
$\mathcal S_{q,L}=L^\perp$ and when $\ell=c=1$, we have $L=\{0\}$.
\end{proof}
\subsection{Quantum-to-classical security reduction}
\label{subsec:common_security_reduction}

We now show that a classical commitment code for $W_{L,q}$ remains secure
when each use of this channel is implemented by the quantum protocol
of Lemma~\ref{lem:stabilizer_code_flagged_weyl_reduction}.

\begin{lemma}
\label{lem:extractable_classical_reduction}
Suppose that $W_{L,q}$ is nonredundant. For any distribution
$P_{\overline X}$ on $\mathcal X_L$, every rate $R<H(\overline X\mid F^\ell,\overline Y)$
per use of $W_{L,q}$ is achievable by the corresponding quantum
protocol.
\end{lemma}

\begin{proof}
Fix a classical commitment code using $N$ copies of $W_{L,q}$. Associate
each classical input $\bar x\in\mathcal X_L$ with the quantum state
$\Psi_{\bar x}^L$, and let Bob perform the measurement from
Lemma~\ref{lem:stabilizer_code_flagged_weyl_reduction}. The resulting
conditional distribution is
$W_{L,q}(f^\ell,\bar y\mid\bar x)
=
q_L(f^\ell,\bar y-\bar x).$ 
Thus, when both parties are honest, the classical and quantum protocols
produce the same joint distribution of
$(F^{\ell N},\overline Y^N)$. Hence, the quantum protocol has the
same correctness error as the classical code.

We next prove hiding.
Conditioned on $F^\ell=f^\ell$ and $Z^\ell=z^\ell$, the channel reveals
$f^\ell$ to Bob and transforms the encoded state according to
$
(W_{z^\ell}\otimes I^R)
\Psi_{\bar x}^L
(W_{z^\ell}^{\dagger}\otimes I^R)
=
\Psi_{\bar x+\pi_L(z^\ell)}^L$, which means that Bob receives the state with label
$\bar y=\bar x+\pi_L(z^\ell)$. Since, by
\eqref{eq:stabilizer_quotient_additive_channel}, the probability of
obtaining the labels $(f^\ell,\bar y)$ is
$W_{L,q}(f^\ell,\bar y\mid\bar x)$, Bob's quantum output for
input $\bar x$ is
\begin{align}
\sum_{f^\ell,\bar y}
W_{L,q}(f^\ell,\bar y\mid\bar x)
|f^\ell\rangle\!\langle f^\ell|
\otimes
\Psi_{\bar y}^L
=
\mathcal P\left(
\sum_{f^\ell,\bar y}
W_{L,q}(f^\ell,\bar y\mid\bar x)
|f^\ell,\bar y\rangle\!\langle f^\ell,\bar y|
\right), \label{eqoneuse}
\end{align}
with $\mathcal P
\left(
|f^\ell,\bar y\rangle\!\langle f^\ell,\bar y|
\right)
\triangleq
|f^\ell\rangle\!\langle f^\ell|
\otimes
\Psi_{\bar y}^L.$ When Alice commits to $s$, the classical commitment code selects
$\overline X^N$ according to
$P_{\overline X^N|S}(\bar x^N|s)$ so that its classical output state is 
\begin{align*}
\rho_s^{(F^\ell\overline Y)^N}
\triangleq
\sum_{\bar x^N}
P_{\overline X^N|S}(\bar x^N|s)
\bigotimes_{j=1}^N
\left(
\sum_{f_j^\ell,\bar y_j}
W_{L,q}(f_j^\ell,\bar y_j\mid\bar x_j)
|f_j^\ell,\bar y_j\rangle
\!\langle f_j^\ell,\bar y_j|
\right),\end{align*}
and, by \eqref{eqoneuse},  Bob's quantum channel
output  is
$\mathcal P^{\otimes N}
\left(
\rho_s^{(F^\ell\overline Y)^N}
\right).$ The hiding property of the classical commitment code gives  $\omega^{(F^\ell\overline Y)^N}$, independent of $s$, such that
$\left\|
\rho_s^{(F^\ell\overline Y)^N}
-
\omega^{(F^\ell\overline Y)^N}
\right\|_1
\leq
\beta_N.$ Fix an arbitrary dishonest receiver strategy
$\mathsf B_{\mathrm{com}}^\star$. Since Bob sends no messages to Alice
during the commit phase, his operations can be represented by a quantum
operation $\mathcal D_{\mathsf B_{\mathrm{com}}^\star}$ applied after the
$N$ channel outputs are produced so that his final state
when Alice commits to $s$ is $\theta_{s,\mathsf B_{\mathrm{com}}^\star}
\triangleq
\mathcal D_{\mathsf B_{\mathrm{com}}^\star}
\left(
\mathcal P^{\otimes N}
\left(
\rho_s^{(F^\ell\overline Y)^N}
\right)
\right).
$
Then, with
$
\tau_{\mathsf B_{\mathrm{com}}^\star}
\triangleq
\mathcal D_{\mathsf B_{\mathrm{com}}^\star}
\left(
\mathcal P^{\otimes N}
\left(
\omega^{(F^\ell\overline Y)^N}
\right)
\right)$, for every $s$, we have
\begin{align*}
\left\|
\theta_{s,\mathsf B_{\mathrm{com}}^\star}
-
\tau_{\mathsf B_{\mathrm{com}}^\star}
\right\|_1
&\overset{(a)}{\leq}
\left\|
\rho_s^{(F^\ell\overline Y)^N}
-
\omega^{(F^\ell\overline Y)^N}
\right\|_1
\\
&\overset{(b)}{\leq}
\beta_N,
\end{align*}
where $(a)$ holds by monotonicity of the trace distance \cite[Eq.~(9.69)]{wilde2013quantum},  $(b)$ holds by the hiding
property of the classical commitment code.

We now prove binding. Fix an arbitrary dishonest quantum commit
strategy. Since Alice receives
no messages from Bob during the commit phase, immediately before the
$N$ channel blocks there is a joint state
$\Gamma^{(A^\ell R)^N R_A}$
of Alice's $N$ channel inputs, Bob's retained registers $R^N$, and
Alice's remaining register $R_A$. For each block, Bob measures the channel output and the corresponding
register $R$ in the basis
$\{|\Phi_{\bar y}^L\rangle\}_{\bar y\in\mathcal X_L}$, obtaining
$\overline Y_j$, and records the flag $F_j^\ell$. Let
$
\mathsf{Ext}^{\mathrm{cl}}_N:
\mathcal X_L^N\to[2^{k_N}]
$
denote the extraction function associated with the classical commitment
code so that
$\mathsf{Ext}^{\mathrm{cl}}_N(\bar x^N)$ is the value extracted by the
classical binding proof when the induced-channel input sequence is
$\bar x^N$. Then, define
\begin{align*}
\Lambda_{\widetilde s,c}^{B_{\mathrm{aux}}E^{\ell N}}
\triangleq
I^{B_{\mathrm{aux}}}
\otimes
\sum_{\substack{f'^{\ell N},z^{\ell N}:\\
\mathsf{Ext}^{\mathrm{cl}}_N\left(
\bar y^N-\pi_L^{\otimes N}(z^{\ell N})
\right)=\widetilde s}}
|f'^{\ell N},z^{\ell N}\rangle
\!\langle f'^{\ell N},z^{\ell N}|,
\end{align*}
which form a POVM that is fixed before Alice chooses her reveal strategy, and whose outcome is
$
\widetilde S
=
\mathsf{Ext}^{\mathrm{cl}}_N(\widetilde X^N)$, with $\widetilde X_j
=
\overline Y_j-\pi_L(Z_j^\ell)$. Fix an arbitrary dishonest reveal strategy represented by a measurement
$\{\Omega_d^{R_A}\}_d$ on $R_A$. Let $D$ denote Alice's
classical reveal data. Then, for
$\bar x^N=(\bar x_1,\ldots,\bar x_N)$, define
$\Psi_{\bar x^N}^L
\triangleq
\bigotimes_{j=1}^N\Psi_{\bar x_j}^L$, and the joint distribution generated by the dishonest quantum strategy
satisfies
\begin{align}
&\Prb\left[
\widetilde X^N=\bar x^N,\,
F^{\ell N}=f^{\ell N},\,
\overline Y^N=\bar y^N,\,
D=d
\right]
\notag\\
&=
\sum_{\substack{z_1^\ell,\ldots,z_N^\ell:\\
\bar x_j=\bar y_j-\pi_L(z_j^\ell),\,j\in[N]}}
\left(
\prod_{j=1}^N
q^{\otimes\ell}(f_j^\ell,z_j^\ell)
\right)
\operatorname{Tr}\left[
\left(
\left[
\bigotimes_{j=1}^N
(W_{z_j^\ell}^{\dagger}\otimes I^R)
\Psi_{\bar y_j}^L
(W_{z_j^\ell}\otimes I^R)
\right]
\otimes
\Omega_d^{R_A}
\right)
\Gamma
\right]
\notag\\
&\overset{(a)}{=}
\sum_{\substack{z_1^\ell,\ldots,z_N^\ell:\\
\bar x_j=\bar y_j-\pi_L(z_j^\ell),\,j\in[N]}}
\left(
\prod_{j=1}^N
q^{\otimes\ell}(f_j^\ell,z_j^\ell)
\right)
\operatorname{Tr}\left[
\left(
\Psi_{\bar x^N}^L
\otimes
\Omega_d^{R_A}
\right)
\Gamma
\right]
\notag\\
&\overset{(b)}{=}
W_{L,q}^N
\left(
f^{\ell N},\bar y^N
\mid
\bar x^N
\right)
\operatorname{Tr}\left[
\left(
\Psi_{\bar x^N}^L
\otimes
\Omega_d^{R_A}
\right)
\Gamma
\right],
\label{eq:quantum_classical_binding_distribution}
\end{align}
where $(a)$ holds because for every term in the sum,
$
(W_{z_j^\ell}^{\dagger}\otimes I^R)
\Psi_{\bar y_j}^L
(W_{z_j^\ell}\otimes I^R)
=
\Psi_{\bar y_j-\pi_L(z_j^\ell)}^L
=
\Psi_{\bar x_j}^L$,  $(b)$ holds by
\eqref{eq:stabilizer_quotient_additive_channel}.

We now construct a dishonest strategy for the classical commitment code, and let $\Prb_{\mathrm{cl}}$ denote probabilities under this strategy.
It selects $(\overline X^N,D)$ according to
\begin{equation}
\Prb_{\mathrm{cl}}[
\overline X^N=\bar x^N,\,
D=d
]
\triangleq
\operatorname{Tr}\left[
\left(
\Psi_{\bar x^N}^L
\otimes
\Omega_d^{R_A}
\right)
\Gamma
\right]
\label{eq:classical_strategy_from_quantum_strategy}
\end{equation}
and sends $\overline X^N$ through $W_{L,q}^N$. Hence, by \eqref{eq:quantum_classical_binding_distribution} and
\eqref{eq:classical_strategy_from_quantum_strategy}
\begin{align*}
\Prb\left[
\widetilde X^N=\bar x^N,\,
F^{\ell N}=f^{\ell N},\,
\overline Y^N=\bar y^N,\,
D=d
\right]
=
\Prb_{\mathrm{cl}}\left[
\overline X^N=\bar x^N,\,
F^{\ell N}=f^{\ell N},\,
\overline Y^N=\bar y^N,\,
D=d
\right],
\end{align*}
hence, with $\widehat S$  Alice's revealed sequence, the binding property of the classical commitment code gives
$\Prb\left[
J=0,\,
\widehat S\neq\widetilde S
\right]
=
\Prb_{\mathrm{cl}}\left[
J=0,\,
\widehat S\neq
\mathsf{Ext}^{\mathrm{cl}}_N(\overline X^N)
\right]
\leq
\nu_N
.$
\end{proof}

\subsection{Achievable rates}
\label{subsec:completion_achievability}
Assume first that $\mathcal S_q=\mathbb F_d^2$. Fix
$R<\min\{H(Z|F),\log_2d+E\}$, and set
$e\triangleq\min\{E/\log_2d,1\}$ and
$c_\ell\triangleq\lfloor e\ell\rfloor$.   Lemma~\ref{lem:conditional_random_stabilizer_hashing} gives
self-orthogonal subspaces $L_\ell$ of dimension $\ell-c_\ell$ such that
$\ell^{-1}H(\pi_{L_\ell}(Z^\ell)|F^\ell)$ converges to
$\min\{H(Z|F),(1+e)\log_2d\}
=\min\{H(Z|F),\log_2d+E\}$. Since $R$ is strictly smaller than this
limit, for all sufficiently large $\ell$ we may choose $L_\ell$ such that
$H(\pi_{L_\ell}(Z^\ell)|F^\ell)>\ell R$. For $\overline X$  uniform on $\mathcal X_{L_\ell}$, the induced channel $W_{L_\ell,q}$ satisfies 
$H(\overline X|F^\ell,\overline Y)
=H(\pi_{L_\ell}(Z^\ell)|F^\ell)>\ell R$ by
Lemma~\ref{lem:stabilizer_code_flagged_weyl_reduction}, and is nonredundant by
Lemma~\ref{lem:stabilizer_induced_channel_nondegeneracy},  therefore,
Theorem~\ref{thm:dmc_extractable_commitment_achievability} and
Lemma~\ref{lem:extractable_classical_reduction} achieve rate $\ell R$
per use of $W_{L_\ell,q}$. Each use of $W_{L_\ell,q}$ requires $\ell$ physical channel uses and
$c_\ell$ maximally entangled qudit pairs. Thus, rate $\ell R$ per induced
channel use equals rate $R$ per physical channel use. For $N$ induced
uses, the preshared state has Schmidt rank $d^{Nc_\ell}$, and its
entanglement rate is
$\log_2(d^{Nc_\ell})/(N\ell)
=(c_\ell/\ell)\log_2d
\leq e\log_2d\leq E$. Hence, every
$R<\min\{H(Z|F),\log_2d+E\}$ is achievable.

Assume now only that
$\operatorname{span}_{\mathbb F_p}\mathcal S_q=\mathbb F_d^2$. Choose $\ell=c=1$, thus, each induced channel use consists of one physical channel
use and one maximally entangled qudit pair, corresponding to entanglement rate $\log_2d$. Moreover,
$\mathcal X_L=\mathbb F_d^2$ and $\pi_L(Z)=Z$. For $\overline X$  uniform on $\mathbb F_d^2$, the induced
channel satisfies
$H(\overline X|F,\overline Y)=H(\pi_L(Z)|F)=H(Z|F)$ by
Lemma~\ref{lem:stabilizer_code_flagged_weyl_reduction}, and is 
 nonredundant by Lemma~\ref{lem:stabilizer_induced_channel_nondegeneracy}, therefore,  Theorem~\ref{thm:dmc_extractable_commitment_achievability}  and
Lemma~\ref{lem:extractable_classical_reduction} achieve every rate 
$R<H(Z|F)$. 
\section{Converse of Theorem \ref{thm:flagged_weyl_capacity}}
\label{sec:converse}

We first prove a  lemma that bounds the information supplied by the
environment of $\mathcal N_q$. 

\begin{lemma}
\label{lem:flagged_weyl_environment_information}
Let $\rho^{RA^nB_0}$ be an arbitrary state, where $B_0$ is an arbitrary
 system held by Bob. Apply
$V_q^{\otimes n}$ from \eqref{eq:flagged_weyl_isometry} to $A^n$, and
denote the resulting state by
$\omega^{RF^nB^nB_0E^n}$. Then,
\begin{equation*}
I(R;E^n|F^nB^nB_0)_\omega
\leq
nH(Z|F).
\end{equation*}
\end{lemma}
\begin{proof}
Let $T$ purify $\rho^{RA^nB_0}$ and set
$G\triangleq F^nB^nB_0$. Then,
\begin{align}
I(R;E^n|G)_\omega
&\overset{(a)}{\leq}
I(RT;E^n|G)_\omega
\notag\\
&\overset{(b)}{=}
H(E^n)_\omega
+
H(RT)_\rho
-
H(G)_\omega
\notag\\
&\overset{(c)}{\leq}
H(F^n,Z^n)
+
H(RT)_\rho
-
H(G)_\omega
\notag\\
&\overset{(d)}{=}
H(F^n,Z^n)
+
H(RT)_\rho
-
H(F^n)
-
\textstyle\sum_{f^n}
P_{F^n}(f^n)
H(B^nB_0|F^n=f^n)_\omega
\notag\\
&\overset{(e)}{\leq}
H(F^n,Z^n)
+
H(RT)_\rho
-
H(F^n)
-
H(A^nB_0)_\rho
\notag\\
&\overset{(f)}{=}
H(F^n,Z^n)-H(F^n)
\notag\\
&=
H(Z^n|F^n)
\notag\\
&\overset{(g)}{=}
nH(Z|F),\notag
\end{align}
where $(a)$ follows from the chain rule and the nonnegativity of
conditional mutual information, $(b)$ holds  because, since $V_q^{\otimes n}$ is an isometry,
applying it to the pure state on $RTA^nB_0$ produces a pure state on
$RTGE^n$, and $I(RT;E^n|G)_\omega
=H(RTG)_\omega+H(E^nG)_\omega-H(G)_\omega-H(RTGE^n)_\omega
=H(E^n)_\omega+H(RT)_\omega-H(G)_\omega
=H(E^n)_\omega+H(RT)_\rho-H(G)_\omega$ by \cite[Th.~11.2.1]{wilde2013quantum},  $(c)$ holds because the rank-one projective measurement
$\{|f^n,z^n\rangle\!\langle f^n,z^n|\}_{f^n,z^n}$ on $E^n$
has outcome distribution $P_{F^nZ^n}$ and therefore outcome entropy
$H(F^n,Z^n)$, so that, by \cite[Th.~11.1.1]{wilde2013quantum}, 
$H(E^n)_\omega\leq H(F^n,Z^n)$, $(d)$
holds because  $F^n$ is classical,  $(e)$ holds because
\begin{align}
H(B^nB_0|F^n=f^n)_\omega
&=
H\left(
\sum_{z^n}
P_{Z^n|F^n}(z^n|f^n)
(W_{z^n}\otimes I^{B_0})
\rho^{A^nB_0}
(W_{z^n}^{\dagger}\otimes I^{B_0})
\right)
\notag\\
&\geq
\sum_{z^n}
P_{Z^n|F^n}(z^n|f^n)
H\left(
(W_{z^n}\otimes I^{B_0})
\rho^{A^nB_0}
(W_{z^n}^{\dagger}\otimes I^{B_0})
\right)
\notag\\
&=
\sum_{z^n}
P_{Z^n|F^n}(z^n|f^n)
H(A^nB_0)_\rho
\notag\\
&=
H(A^nB_0)_\rho,
\end{align}
where the first inequality holds by \cite[Prop.~11.1.4]{wilde2013quantum}, and the second equality holds 
by unitary
invariance of entropy
\cite[Prop.~11.1.5]{wilde2013quantum}, $(f)$ holds by \cite[Th.~11.2.1]{wilde2013quantum} since the state on
$RTA^nB_0$ is pure,  $(g)$ holds because $(F^n,Z^n)$ is
memoryless.
\end{proof}

We now consider an arbitrary noninteractive
$(n,k,K,\alpha,\beta,\gamma)$ commitment protocol. Let $S$ be uniform on
$\{0,1\}^k$, and let Alice execute the honest commit strategy for $S$. Consider a dishonest receiver  who stores the
raw channel outputs and whose view is
$V\triangleq F^nB^nB_0$, where $B_0$ is Bob's share of the initial
entangled state. If $C_{\mathrm{com}}=c$,  the binding
condition gives
a measurement  
$\{\Lambda_{\widetilde s,c}^{B_{\mathrm{aux}}E^n}\}_{\widetilde s}$ so that
$\Lambda_{\widetilde s}
=\sum_c|c\rangle\!\langle c|^{C_{\mathrm{com}}}
\otimes\Lambda_{\widetilde s,c}^{B_{\mathrm{aux}}E^n}$. Bob's honest commit-phase record is obtained from $V$ by a quantum
operation, which composed with the measurement supplied by the
binding condition gives a measurement of $VE^n$ whose outcome is denoted
by $\widetilde S$. Then, we have
\begin{align}
\left(
1-\beta-\Prb[\widetilde S\neq S]
\right)\frac{k}{n}
&=
\frac{
H(S)-\beta k-\Prb[\widetilde S\neq S]k
}{n}
\notag\\
&=
\frac{
I(S;V)
+
I(S;E^n|V)
+
H(S|VE^n)
-
\beta k
-
\Prb[\widetilde S\neq S]k
}{n}
\notag\\
&\overset{(a)}{\leq}
\frac{
\left(1+\frac{\beta}{2}\right)
h_2\left(\frac{\beta}{2+\beta}\right)
+
I(S;E^n|V)
+
H(S|VE^n)
-
\Prb[\widetilde S\neq S]k
}{n}
\notag\\
&\overset{(b)}{\leq}
H(Z|F)
+
\frac{
\left(1+\frac{\beta}{2}\right)
h_2\left(\frac{\beta}{2+\beta}\right)
+
H(S|VE^n)
-
\Prb[\widetilde S\neq S]k
}{n}
\notag\\
&\overset{(c)}{\leq}
H(Z|F)
+
\frac{1}{n}
\left(1+\frac{\beta}{2}\right)
h_2\left(\frac{\beta}{2+\beta}\right)
+
\frac{1}{n}
h_2\left(\Prb[\widetilde S\neq S]\right), 
\label{eq:flagged_weyl_finite_block_converse}
\end{align}
where $(a)$ holds because the
hiding condition gives a state $\tau^V$, independent of $s$, such that
$\|\theta_s^V-\tau^V\|_1\leq\beta$ with $\theta_s^V$  the receiver's view when $S=s$, so that  $\|\rho^{SV}-\sigma^{SV}\|_1\leq\beta$ with 
$\rho^{SV}\triangleq
2^{-k}\sum_s|s\rangle\!\langle s|^S\otimes\theta_s^V$ and
$\sigma^{SV}\triangleq
2^{-k}\sum_s|s\rangle\!\langle s|^S\otimes\tau^V$,  and, by
\cite[Th.~11.10.3]{wilde2013quantum},  
$I(S;V)\leq
\beta k+(1+\beta/2)h_2(\beta/(2+\beta))$, $(b)$ holds because in the absence of commit-phase interaction Alice's
channel inputs and Bob's initial system $B_0$ can be represented jointly
before the channel uses and Lemma
\ref{lem:flagged_weyl_environment_information}, applied with $R=S$,
 gives $I(S;E^n|V)\leq nH(Z|F)$,  $(c)$ follows because
$\widetilde S$ is obtained by measuring $VE^n$, so data processing and
Fano's inequality give
$H(S|VE^n)\leq
h_2(\Prb[\widetilde S\neq S])
+\Prb[\widetilde S\neq S]k$.

Let Alice  apply the  honest
reveal map to $S$, and let $\widehat S$ be Bob's output. Then, we have
\begin{align*}
\Prb[\widetilde S\neq S]
&\overset{(a)}{\leq}
\Prb[(\widehat S=S,J=0)^c]
+
\Prb[J=0,\widehat S\neq\widetilde S]
\\
&\overset{(b)}{\leq}
\alpha+\gamma, \numberthis \label{eq:flagged_weyl_finite_block_converseb}
\end{align*}
where $(a)$ holds by the union bound, $(b)$ follows from
correctness and binding. Hence, 
\eqref{eq:flagged_weyl_finite_block_converse} and \eqref{eq:flagged_weyl_finite_block_converseb} give
$\limsup_{n\to\infty}k_n/n\leq H(Z|F)$.

We now make use of the entanglement constraint. We have
\begin{align}
\left(
1-\Prb[\widetilde S\neq S]
\right)\frac{k}{n}
&=
\frac{
H(S)-\Prb[\widetilde S\neq S]k
}{n}
\notag\\
&=
\frac{
I(S;VE^n)
+
H(S|VE^n)
-
\Prb[\widetilde S\neq S]k
}{n}
\notag\\
&\overset{(a)}{\leq}
\frac{
H(VE^n)
+
H(S|VE^n)
-
\Prb[\widetilde S\neq S]k
}{n}
\notag\\
&\overset{(b)}{\leq}
\frac{
H(VE^n)
+
h_2\left(\Prb[\widetilde S\neq S]\right)
}{n}
\notag\\
&\overset{(c)}{\leq}
\log_2d
+
\frac{\log_2K}{n}
+
\frac{1}{n}
h_2\left(\Prb[\widetilde S\neq S]\right),
\label{eq:finite_entanglement_dimension_bound}
\end{align}
where $(a)$ follows because, since $S$ is classical,
$H(VE^n|S)=\sum_sP_S(s)H(\rho_s^{VE^n})\geq0$ by the nonnegativity of
quantum entropy, $(b)$ holds as in \eqref{eq:flagged_weyl_finite_block_converse}, $(c)$ holds because Bob's share $B_0$ of the
initial pure state has support dimension at most $K$, so that  the joint
input on $A^nB_0$ has support dimension at most $d^nK$, and  since the channel isometry maps $A^nB_0$ to
$VE^n=F^nB^nB_0E^n$ without changing its support dimension, we have
$H(VE^n)\leq\log_2(d^nK)=n\log_2d+\log_2K$. Hence, for an achievable sequence with entanglement rate $E$, we have
$\limsup_{n\to\infty}n^{-1}\log_2K_n\leq E$, so that 
\eqref{eq:flagged_weyl_finite_block_converseb} and 
\eqref{eq:finite_entanglement_dimension_bound} give
$\limsup_{n\to\infty}k_n/n\leq\log_2d+E$.

\section{Interactive Protocols}
\label{sec:interactive}
\label{sec:interactive_erasure}

The converse in Section \ref{sec:converse} assumes that Bob sends no
messages to Alice during the commit phase. In Section~\ref{subsec:interactive_protocols}, we extend the protocol definition to
allow authenticated classical communication in both directions. In
Section \ref{subsec:interactive_dimension_bound}, we prove an upper bound
for every channel with input dimension $d$ and identify when it is attained
by our noninteractive construction. Finally, in Section
\ref{subsec:interactive_erasure_channel},  for
the quantum erasure channel, we determine its interactive capacity for every
entanglement rate.
 \subsection{Interactive commitment protocols}
\label{subsec:interactive_protocols}

Let $\mathcal N:A\to B$ be a quantum channel with
$\dim\mathcal H_A=d$, and fix an isometric extension
$V_{\mathcal N}:A\to BW$.

\begin{definition}
\label{def:interactive_commitment_protocol}
An $(n,k,K)$ interactive classical string commitment protocol over
$\mathcal N$ is defined as in Section~\ref{sec:problem_statement},
including an initial pure state of Schmidt rank at most $K$, except that
the parties may exchange finitely many authenticated classical messages
before, between, and after the $n$ channel uses. At the $i$th use, Alice
sends a $d$-dimensional system $A_i$ through $\mathcal N$. Each party may apply arbitrary operations to its current classical and
quantum registers. At the end of the commit phase,
$C_{\mathrm{com}}$ contains the public transcript and Bob's retained
classical outcomes, while $B_{\mathrm{aux}}$ contains his retained
quantum systems. The reveal phase and the correctness, hiding, and binding requirements
are those of Section~\ref{sec:problem_statement}, with the channel
environment $E^n$ replaced by $W^n$.
\end{definition}

\begin{definition}
Fix $E\geq0$. A rate $R\geq0$ is  achievable with
entanglement rate $E$ if there exists a sequence of
$(n,k_n,K_n,\alpha_n,\beta_n,\gamma_n)$ interactive protocols such that
\begin{equation*}
\liminf_{n\to\infty}\frac{k_n}{n}\geq R,
\qquad
\limsup_{n\to\infty}\frac{\log_2K_n}{n}\leq E,
\qquad
\lim_{n\to\infty}
\alpha_n=\lim_{n\to\infty}
\beta_n=\lim_{n\to\infty}
\gamma_n=0.
\end{equation*}
The supremum of all such rates is denoted by
$C_{\mathrm{com,int}}(E;\mathcal N)$. We write
$C_{\mathrm{com,int}}(\mathcal N)$ when there is no preshared entanglement, and
$C_{\mathrm{com,int}}^{\mathrm{EA}}(\mathcal N)$ when the preshared
entanglement is unrestricted.
\end{definition}

\subsection{Universal converse bound}
\label{subsec:interactive_dimension_bound}

\begin{theorem}
\label{thm:universal_interactive_dimension_converse}
For every quantum channel $\mathcal N:A\to B$ with
$\dim\mathcal H_A=d$ and every $E\geq0$,
\begin{equation}
C_{\mathrm{com,int}}(E;\mathcal N)
\leq
\log_2d+E.
\label{eq:universal_interactive_dimension_converse}
\end{equation}
\end{theorem}

\begin{proof}
See Appendix \ref{App_convi}.\end{proof}

As stated next, Theorem \ref{thm:universal_interactive_dimension_converse} determines the interactive capacity whenever the
noninteractive construction of Theorem \ref{thm:flagged_weyl_capacity} attains $\log_2d+E$.

\begin{corollary}
\label{cor:interactive_saturated_flagged_weyl}
Let $\mathcal N_q$ be the channel in
\eqref{eq:flagged_weyl_channel}, suppose that
$\mathcal S_q=\mathbb F_d^2$, and let $0\leq E\leq\log_2d$. If
 $H(Z|F)\geq\log_2d+E,$
then
\begin{equation*}
C_{\mathrm{com,int}}(E;\mathcal N_q)
=
\log_2d+E.
\end{equation*}
\end{corollary}

For qubit Pauli channels, we obtain the following result from Corollaries \ref{cor:qubit_pauli_capacity} and  \ref{cor:interactive_saturated_flagged_weyl}.

\begin{corollary}
\label{cor:interactive_saturated_pauli}
Let $\mathcal P_q$ be as in Corollary \ref{cor:qubit_pauli_capacity} such that all three coefficients
in \eqref{eq:qubit_pauli_nondegeneracy} are nonzero. If $0\leq E\leq1$ and $
H(q_I,q_X,q_Y,q_Z)\geq1+E$, then
\begin{equation*}
C_{\mathrm{com,int}}(E;\mathcal P_q)=1+E.
\end{equation*}
\end{corollary}

For the depolarizing channel, we obtain the following result from Corollaries \ref{cor:depolarizing_capacity} and  \ref{cor:interactive_saturated_flagged_weyl}.
\begin{corollary}
\label{cor:interactive_saturated_paulil}
For the depolarizing channel $\mathcal D_p$ in
\eqref{eq:depolarizing_channel} and $0\leq E<1$, we have
\begin{equation*}
C_{\mathrm{com,int}}(E;\mathcal D_p)=1+E,
\qquad
p_\star(E)\leq p<1,
\end{equation*}
where $p_\star(E)$ is the unique solution
in $(0,1)$ of
$
H\left(
1-\frac{3p_\star(E)}4,
\frac{p_\star(E)}4,
\frac{p_\star(E)}4,
\frac{p_\star(E)}4
\right)
=
1+E.$
\end{corollary}

\subsection{The quantum erasure channel}
\label{subsec:interactive_erasure_channel}
For the quantum erasure channel, we obtain its interactive capacity for every
entanglement rate, which reveals that interaction does not increase the capacity.

\begin{theorem}
\label{thm:interactive_erasure_capacity}
For $0<\epsilon<1$ and $E\geq0$, we have
\begin{equation*}
C_{\mathrm{com,int}}(E;\mathcal E_{\epsilon,d})
= C_{\mathrm{com}}(E;\mathcal E_{\epsilon,d})
=\min\{2\epsilon\log_2d,\log_2d+E\}.
\end{equation*}
\end{theorem}

\begin{proof}
See Appendix \ref{converas}.
\end{proof}

\section{Concluding Remarks}
\label{sec:conclusion}

We determined the noninteractive classical string commitment capacity of
the channels $\mathcal N_q$ satisfying
$\mathcal S_q=\mathbb F_d^2$ for every entanglement rate $E$, namely,
$C_{\mathrm{com}}(E;\mathcal N_q)
=\min\{H(Z|F),\log_2d+E\}$.
The capacity with unrestricted preshared entanglement remains equal to
$H(Z|F)$ under the weaker condition
$\operatorname{span}_{\mathbb F_p}\mathcal S_q=\mathbb F_d^2$. These
results apply, in particular, to quantum erasure, depolarizing, dephrasure,
and Pauli qudit channels satisfying the corresponding conditions. For interactive protocols, we proved the upper bound
$\log_2d+E$, which matches the noninteractive capacity whenever
$H(Z|F)\geq\log_2d+E$. For the quantum erasure channel, an additional
 bound gives the  interactive capacity, showing that interaction does not
increase the capacity in this case. It remains open whether interaction can increase the
capacity in general. Open
directions include finite-blocklength bounds and  efficient
constructions.

\appendices

\section{Proof of Theorem
\ref{thm:dmc_extractable_commitment_achievability}}
\label{App_th}

The following lemma is established in
\cite[Prop.~7 and~8]{winter2003commitment}.

\begin{lemma}[{\cite[Prop.~7 and~8]{winter2003commitment}}]
\label{lem:classical_commitment_code}
Let $W:\mathcal X\to\mathcal Y$ be a finite nonredundant classical channel.
Fix an input distribution $P_X$ and a rate $R<H(X|Y)$, where the entropy is
computed under $P_{XY}(x,y)=P_X(x)W(y|x)$. Then, there exist $\sigma>0$
and sequences of
integers $(k_N)_{N\in\mathbb N}$ and $(L_N)_{N\in\mathbb N}$ such that, for
every sufficiently large $N$, there exist maps $\phi_N:
[2^{k_N}]\times[L_N]
\longrightarrow
\mathcal X^N$
and sets
$\mathcal B_{s,r}^{(N)}\subseteq\mathcal Y^N$,
$s\in[2^{k_N}]$, $r\in[L_N]$, satisfying the following properties.
\begin{itemize}
\item The message rate satisfies
\begin{equation*}
\liminf_{N\to\infty}\frac{k_N}{N}\geq R.
\end{equation*}

\item For all $(s,r)\neq(s',r')$,
\begin{equation}
d_{\mathrm H}
\left(
\phi_N(s,r),
\phi_N(s',r')
\right)
\geq
2\sigma N.
\label{eq:classical_commitment_code_distance}
\end{equation}

\item The output distributions satisfy
\begin{equation}
\max_{s\in[2^{k_N}]}
\left\|
\frac{1}{L_N}
\sum_{r=1}^{L_N}
W^N
\left(
\,\cdot\mid\phi_N(s,r)
\right)
-
Q_Y^{\otimes N}
\right\|_1
\leq
\beta_N,
\label{eq:classical_commitment_code_hiding}
\end{equation}
where $\lim_{N\to\infty}\beta_N=0$ and
$
Q_Y(y)
\triangleq
\sum_{x\in\mathcal X}
P_X(x)W(y|x).$

\item We have
\begin{align}
W^N
\left(
\mathcal B_{s,r}^{(N)}
\middle|
\phi_N(s,r)
\right)
&\geq
1-\alpha_N,
\label{eq:classical_commitment_code_completeness}\\
\sup_{\substack{u^N\in\mathcal X^N:\\
d_{\mathrm H}(u^N,\phi_N(s,r))\geq\sigma N}}
W^N
\left(
\mathcal B_{s,r}^{(N)}
\middle|
u^N
\right)
&\leq
\nu_N,
\label{eq:classical_commitment_code_soundness}
\end{align}
where $\lim_{N\to\infty}\alpha_N=0=\lim_{N\to\infty}\nu_N$.
\end{itemize}
\end{lemma}

Fix $R<\max_{P_X}H(X|Y)$, and choose an input distribution $P_X$ such that
$R<H(X|Y)$. Using Lemma
\ref{lem:classical_commitment_code} for this distribution, Alice commits to
$s\in\{0,1\}^{k_N}$ by sampling $r$ uniformly from $[L_N]$ and sending $X^N=\phi_N(s,r)$
through $W^N$. To reveal, Alice sends $(s,r)$. Bob sets
$\widehat S=s$ and $J=0$ if
$Y^N\in\mathcal B_{s,r}^{(N)}$, otherwise he sets $J=1$.
Correctness follows from
\eqref{eq:classical_commitment_code_completeness}.

We next verify hiding against an arbitrary dishonest Bob. Since the protocol
contains no commit-phase public communication, Bob's strategy is described
by a quantum channel $\mathcal D:
Y^N
\longrightarrow
C_{\mathrm{com}}B_{\mathrm{aux}}.$ 
For an honest commitment to $s$, the state of the classical channel-output
register is
\begin{equation*}
\rho_s^{Y^N}
\triangleq
\sum_{y^N}
\left[
\frac{1}{L_N}
\sum_{r=1}^{L_N}
W^N
\left(
y^N\middle|\phi_N(s,r)
\right)
\right]
|y^N\rangle\!\langle y^N|
\end{equation*}
so that 
$\theta_s^{C_{\mathrm{com}}B_{\mathrm{aux}}}
=\mathcal D(\rho_s^{Y^N})$. Then, define $
\omega^{Y^N}
\triangleq
\sum_{y^N}
Q_Y^{\otimes N}(y^N)
|y^N\rangle\!\langle y^N|$ 
and
$\tau^{C_{\mathrm{com}}B_{\mathrm{aux}}}
\triangleq\mathcal D(\omega^{Y^N})$. Then,
\begin{align*}
\left\|
\theta_s^{C_{\mathrm{com}}B_{\mathrm{aux}}}
-
\tau^{C_{\mathrm{com}}B_{\mathrm{aux}}}
\right\|_1
&\overset{(a)}\leq
\left\|
\rho_s^{Y^N}
-
\omega^{Y^N}
\right\|_1
\notag\\
&=
\left\|
\frac{1}{L_N}
\sum_{r=1}^{L_N}
W^N
\left(
\,\cdot\mid\phi_N(s,r)
\right)
-
Q_Y^{\otimes N}
\right\|_1
\notag\\
&\overset{(b)}\leq
\beta_N,
\end{align*}
where $(a)$ holds by monotonicity of the trace distance, e.g., \cite[Eq. (9.69)]{wilde2013quantum}, $(b)$ holds by
\eqref{eq:classical_commitment_code_hiding}.

We now verify binding. Consider an arbitrary dishonest commit strategy,
and let $U^N$ denote the  classical input transmitted through $W^N$.
For $u^N\in\mathcal X^N$, define
\begin{equation*}
\mathcal M(u^N)
\triangleq
\left\{
(s,r)\in[2^{k_N}]\times[L_N]:
d_{\mathrm H}
\bigl(
u^N,\phi_N(s,r)
\bigr)
<
\sigma N
\right\}.
\end{equation*}
By \eqref{eq:classical_commitment_code_distance},
$|\mathcal M(u^N)|\leq1$. Fix $s_0\in[2^{k_N}]$ and define
\begin{equation*}
\widetilde S(u^N)
\triangleq
\begin{cases}
\widetilde s,
&
\text{if }
\mathcal M(u^N)=\{(\widetilde s,\widetilde r)\},
\\[1mm]
s_0,
&
\text{if }
\mathcal M(u^N)=\varnothing,
\end{cases}
\end{equation*}
which does not depend
on Alice's reveal strategy.

Let $(S_{\mathrm{rev}},R_{\mathrm{rev}})$ be Alice's revealed pair. Then,
\begin{align*}
\Prb\left[
J=0,\,
\widehat S\neq\widetilde S
\right]
\notag
&=
\sum_{\substack{u^N,s,r:\\
s\neq\widetilde S(u^N)}}
\Prb\left[
U^N=u^N,\,
S_{\mathrm{rev}}=s,\,
R_{\mathrm{rev}}=r
\right]
W^N
\left(
\mathcal B_{s,r}^{(N)}
\middle|
u^N
\right)
\notag\\
&\leq
\nu_N
\sum_{\substack{u^N,s,r:\\
s\neq\widetilde S(u^N)}}
\Prb\left[
U^N=u^N,\,
S_{\mathrm{rev}}=s,\,
R_{\mathrm{rev}}=r
\right]
\notag\\
&\leq
\nu_N,
\end{align*}
where the first inequality holds by 
\eqref{eq:classical_commitment_code_soundness}. Indeed, if
$\mathcal M(u^N)=\varnothing$, then every revealed codeword is at distance
at least $\sigma N$ from $u^N$ by definition of $\mathcal M(u^N)$, and if
$\mathcal M(u^N)=\{(\widetilde s,\widetilde r)\}$ and
$s\neq\widetilde s$, then the triangle inequality and
\eqref{eq:classical_commitment_code_distance} give
$d_{\mathrm H}
\bigl(
u^N,\phi_N(s,r)
\bigr)
\geq
d_{\mathrm H}
\bigl(
\phi_N(\widetilde s,\widetilde r),
\phi_N(s,r)
\bigr)
-
d_{\mathrm H}
\bigl(
u^N,\phi_N(\widetilde s,\widetilde r)
\bigr)
>
2\sigma N-\sigma N
=
\sigma N.$

 \section{Proof of Lemma \ref{lem:self_orthogonal_joint_eigenspace}} \label{App_lem1}

For $v,w\in L$,   $L\subseteq L^\perp$ gives $[v,w]=0$, hence, by 
\cite[Lem.~5]{ketkar2006nonbinary}, 
$
W_vW_w=W_wW_v,
$
so that the operators $\{W_v:v\in L\}$  can be
simultaneously diagonalized. Choose a common unit eigenvector
$|\varphi\rangle$ and $\{\lambda_v:v\in L\}$ such that
$
W_v|\varphi\rangle=\lambda_v|\varphi\rangle
$
for every $v\in L$ and define 
$
S_v\triangleq\lambda_v^{-1}W_v
$. $\{S_v:v\in L\}$ is an Abelian group of order
$|L|$~since
\begin{align}
S_vS_w
&=
\lambda_v^{-1}\lambda_w^{-1}W_vW_w
\notag\\
&\overset{(a)}{=}
\lambda_v^{-1}\lambda_w^{-1}
\zeta(v,w)W_{v+w}
\notag\\
&\overset{(b)}{=}
\lambda_{v+w}^{-1}W_{v+w}
\notag\\
&=
S_{v+w}, \notag
\end{align}
where $(a)$ holds by \cite[Eq.~(1)]{ketkar2006nonbinary} for some $\zeta(v,w)$,   $(b)$ holds by applying
$
W_vW_w=\zeta(v,w)W_{v+w}
$
to the common eigenvector $|\varphi\rangle$. 

Then, define
$
\mathcal Q_L
\triangleq
\left\{
|\psi\rangle:
S_v|\psi\rangle=|\psi\rangle, 
  \forall v\in L
\right\}
$
and 
$
P_L
\triangleq
\frac{1}{|L|}
\sum_{v\in L}S_v.
$
Since $S_v^\dagger=S_{-v}$, we have $P_L^\dagger=P_L$. Moreover,
$
P_L^2
=
\frac{1}{|L|^2}
\sum_{v,w\in L}S_{v+w}
=
\frac{1}{|L|}
\sum_{u\in L}S_u
=
P_L
$. For every $w\in L$,
$
S_wP_L
=
\frac{1}{|L|}
\sum_{v\in L}S_{w+v}
=
P_L,
$
so every vector in the image of $P_L$ belongs to $\mathcal Q_L$.
Conversely, if $|\psi\rangle\in\mathcal Q_L$, then
$
P_L|\psi\rangle=|\psi\rangle.
$
Thus, $P_L$ is the orthogonal projector onto $\mathcal Q_L$.

Finally, we have
\begin{align}
\dim\mathcal Q_L
& = 
\operatorname{Tr}(P_L)
\notag\\
&=
\frac{1}{|L|}
\sum_{v\in L}\operatorname{Tr}(S_v)
\notag\\
&\overset{(a)}{=}
\frac{d^\ell}{|L|}
\notag\\
&\overset{(b)}{=}
\frac{d^\ell}{d^{\ell-c}}
\notag\\
&=
d^c, \notag
\end{align}
where  $(a)$ holds because  
$
\operatorname{Tr}(S_v)=0
$
for every $v\neq0$, and
$
\operatorname{Tr}(S_0)=\operatorname{Tr}(I)=d^\ell$, $(b)$ holds because
$
|L|
=
d^{\dim_{\mathbb F_d}L}
=
d^{\ell-c}.
$

\section{Proof of Theorem \ref{thm:universal_interactive_dimension_converse}} \label{App_convi}

The following lemma holds by Section \ref{sec:converse},   using only hiding, correctness, binding, data processing,
and Fano's inequality.
\begin{lemma}
\label{lem:interactive_information_reduction}
Fix an $(n,k,K,\alpha,\beta,\gamma)$ interactive protocol, let $S$ be
uniform on $\{0,1\}^k$, and let Alice honestly commit to $S$. Set
$V\triangleq C_{\mathrm{com}}B_{\mathrm{aux}}$. There is a measurement of
$VW^n$ with outcome $\widetilde S$ such that
$\Prb[\widetilde S\neq S]\leq\alpha+\gamma$ and
\begin{align*}
\left(
1-\beta-\Prb[\widetilde S\neq S]
\right)k
&\leq
\left(1+\frac{\beta}{2}\right)
h_2\left(\frac{\beta}{2+\beta}\right)
+I(S;W^n|V)
+h_2\left(\Prb[\widetilde S\neq S]\right).
\end{align*}
\end{lemma}

We will also use the following lemma.
\begin{lemma}
\label{lem:interactive_support_dimension}
Consider an $(n,k,K)$ interactive protocol over $\mathcal N$. For
$i\in\{0,1,\ldots,n\}$, let $C_i$ be Bob's classical record after the
first $i$ channel uses and the classical communication following them. Let
$Q_i$ contain $W^i$, Bob's quantum registers, and the auxiliary systems
produced by his operations before they are discarded. For
$P_{S,C_i}(s,c_i)>0$, let $\rho_{s,c_i}^{Q_i}$ be the state of $Q_i$
conditioned on $S=s$ and $C_i=c_i$. For every $c_i$ with
$P_{C_i}(c_i)>0$, there is a subspace
$\mathcal K_{c_i}\subseteq\mathcal H_{Q_i}$, independent of $s$, such that $
\dim\mathcal K_{c_i}\leq Kd^i$, and $\supp(\rho_{s,c_i}^{Q_i})\subseteq\mathcal K_{c_i}$,
for every $s$ with $P_{S,C_i}(s,c_i)>0$.
\end{lemma}

\begin{proof}
We first show that classical communication does not increase the dimension
of a subspace containing the states of $Q_i$. Suppose that, before a
message is sent, all these states are supported on a subspace $\mathcal K$ of dimension
at most $D$. If Alice sends a message $M$, let $\sigma_{s,m}^{Q_i}$ be the subnormalized state of
$Q_i$ conditioned  on $M=m$. Because Alice produces $M$ by acting only on
registers outside $Q_i$, summing over all possible message values gives
$\sum_m\sigma_{s,m}^{Q_i}=\rho_s^{Q_i}$. Consequently,
$0\preceq\sigma_{s,m}^{Q_i}\preceq\rho_s^{Q_i}$, which implies
$\supp(\sigma_{s,m}^{Q_i})\subseteq\supp(\rho_s^{Q_i})
\subseteq\mathcal K$. If Bob sends a message $M$, then, for each possible value $m$, his local
processing can be implemented by the linear map
$V_m|\phi\rangle\triangleq\sum_{\lambda}|\lambda\rangle\otimes
M_{m,\lambda}|\phi\rangle$ so that, if the state of $Q_i$ is supported on
$\mathcal K$ before the message is sent, then, conditioned on $M=m$, it is
supported on $V_m(\mathcal K)$ with
$\dim V_m(\mathcal K)\leq\dim\mathcal K$. Applying this argument to each message shows that classical
communication does not increase the dimension bound.

We now proceed by induction over $i$. Before the first channel use, Bob's share of
the initial state is supported on a subspace of dimension at most $K$
because the state has Schmidt rank at most $K$. Adding fixed ancillas,
applying Bob's operations, and exchanging the initial classical messages
do not increase this bound. Thus, the claim holds for $i=0$.

Suppose that the claim holds after $i-1$ channel uses, and fix
$C_{i-1}=c_{i-1}$. For every $s$, the state of $Q_{i-1}$ is supported on
$\mathcal K_{c_{i-1}}$. Since $\dim\mathcal H_{A_i}=d$, the joint state of
$A_iQ_{i-1}$ is supported on
$\mathcal H_{A_i}\otimes\mathcal K_{c_{i-1}}$, whose dimension satisfies $\dim(\mathcal H_{A_i}\otimes\mathcal K_{c_{i-1}})
=
d\,\dim\mathcal K_{c_{i-1}}
\leq
d\,Kd^{i-1}
=
Kd^i.$  The channel isometry $V_{\mathcal N}:A_i\to B_iW_i$ maps this subspace to
one of the same dimension. Bob's subsequent operations and the following
classical messages cannot increase this dimension. Therefore, for every
resulting value $c_i$, there is a subspace $\mathcal K_{c_i}$ of dimension
at most $Kd^i$ containing the support of $\rho_{s,c_i}^{Q_i}$ for every
$s$. 
\end{proof}

Fix an $(n,k,K,\alpha,\beta,\gamma)$ interactive protocol, and use $S$,
$V$, and $Q_n$ from Lemmas \ref{lem:interactive_information_reduction}
and \ref{lem:interactive_support_dimension}. Then, we
have
\begin{align}
I(S;W^n|V)
&=
I(S;W^n|C_{\mathrm{com}}B_{\mathrm{aux}})
\notag\\
&\overset{(a)}{\leq}
I(S;B_{\mathrm{aux}}W^n|C_{\mathrm{com}})
\notag\\
&\overset{(b)}{\leq}
I(S;Q_n|C_{\mathrm{com}})
\notag\\
&=
\sum_cP_{C_{\mathrm{com}}}(c)
I(S;Q_n|C_{\mathrm{com}}=c)
\notag\\
&\overset{(c)}{\leq}
\sum_cP_{C_{\mathrm{com}}}(c)
H(Q_n|C_{\mathrm{com}}=c)
\notag\\
&\overset{(d)}{\leq}
\sum_cP_{C_{\mathrm{com}}}(c)
\log_2\dim\mathcal K_c
\notag\\
&\overset{(e)}{\leq}
n\log_2d+\log_2K,
\label{eq:interactive_environment_dimension_bound}
\end{align}
where $(a)$ holds by the chain rule and the nonnegativity of conditional
mutual information, $(b)$ follows from data processing, e.g., \cite[Eq.~(11.207)]{wilde2013quantum},  because $B_{\mathrm{aux}}W^n$ is obtained
from $Q_n$ by discarding systems,  $(c)$ holds
because $S$ is classical and quantum entropy is nonnegative, $(d)$ holds because, conditioned on
$C_{\mathrm{com}}=c$, the state of $Q_n$ is supported on
$\mathcal K_c$ by Lemma
\ref{lem:interactive_support_dimension}, and therefore
$H(Q_n|C_{\mathrm{com}}=c)\leq\log_2\dim\mathcal K_c$ by
\cite[Prop.~11.1.3]{wilde2013quantum}, $(e)$ holds by Lemma
\ref{lem:interactive_support_dimension}.

Finally, combining Lemma \ref{lem:interactive_information_reduction} and
\eqref{eq:interactive_environment_dimension_bound} yields  \eqref{eq:universal_interactive_dimension_converse}.

\section{Proof of Theorem \ref{thm:interactive_erasure_capacity}} \label{converas}
Achievability follows from Corollary \ref{cor:qudit_erasure_capacity}. 
For the converse,  Theorem
\ref{thm:universal_interactive_dimension_converse} gives the bound
$\log_2d+E$, hence, it only remains to prove the bound $2\epsilon\log_2d$. Fix an $(n,k,K,\alpha,\beta,\gamma)$ interactive
protocol, and use $S$ and $V$ from Lemma~\ref{lem:interactive_information_reduction}. Let $G_i=1$ when the $i$th
input is erased and $G_i=0$ otherwise. Bob can determine $G^n$ from the
channel outputs, so we include it in his view and write $V=G^nV'$. Then, we have
\begin{align}
I(S;W^n|V)
&=
I(S;W^n|G^nV')
\notag\\
&\overset{(a)}{=}
\sum_{g^n}
P_{G^n}(g^n)
I(S;W^n|V',G^n=g^n)
\notag\\
&=
\sum_{g^n}
P_{G^n}(g^n)
\left[
H(W^n|V',G^n=g^n)
-
H(W^n|SV',G^n=g^n)
\right]
\notag\\
&\overset{(b)}{\leq}
\sum_{g^n}
P_{G^n}(g^n)
\left[
\left(\sum_{i=1}^n g_i\right)\log_2d
-
\left(-\left(\sum_{i=1}^n g_i\right)\log_2d\right)
\right]
\notag\\
&=
2\sum_{g^n}
P_{G^n}(g^n)
\left(\sum_{i=1}^n g_i\right)\log_2d
\notag\\
&=
2\sum_{i=1}^n
\E[G_i]\log_2d
\notag\\
&\overset{(c)}{=}
2\epsilon n\log_2d,
\label{eq:interactive_erasure_environment_bound}
\end{align}
where $(a)$ holds because $G^n$ is classical,  $(b)$ holds because the isometric
extension of $\mathcal E_{\epsilon,d}^{\otimes n}$ shows that, conditioned on $G^n=g^n$, the state of $W^n$ is
supported on a space of dimension $d^{\sum_{i=1}^n g_i}$, hence,
\cite[Th.~11.5.1]{wilde2013quantum} bounds each conditional entropy in
absolute value by $\left(\sum_{i=1}^n g_i\right)\log_2d$, $(c)$
holds because $\E[G_i]=\epsilon$ for every channel use.

Finally, combining Lemma
\ref{lem:interactive_information_reduction} and
\eqref{eq:interactive_erasure_environment_bound} yields
$\limsup_{n\to\infty}k_n/n\leq2\epsilon\log_2d$. 
\bibliographystyle{ieeetr}
\bibliography{bib}

\end{document}